\documentclass[12pt]{article}
\usepackage{fullpage,amsmath,amsfonts,amssymb,amsthm,epsfig,color, graphicx, natbib, appendix, mathtools, setspace}
\usepackage{hyperref}
\usepackage{xcolor}
\newcommand{\bb}{\mathbb}

\theoremstyle{plain}
\newtheorem{Theorem}{Theorem}
\newtheorem{Cor}{Corollary}
\newtheorem{Prop}{Proposition}
\newtheorem{lemma}{Lemma}
\newtheorem*{Theorem*}{Theorem}

\theoremstyle{definition}
\newtheorem{definition}{Definition}
\newtheorem{example}{Example}
\newtheorem{assumption}{Assumption}

\theoremstyle{remark}

\title{A Theory of Network Games Part I:  Utility Representations\thanks{We are grateful to Krishna Dasaratha, Ben Brooks, Navin Kartik, Elliot Lipnowski, Stephen Morris, Doron Ravid, Peter Wakker and seminar participants at Boston University, UC Berkeley and  Yale University for helpful comments and questions.  Evan Sadler gratefully acknowledges support from the Alfred P. Sloan Foundation.}}
\author{Joseph Root\thanks{University of Chicago -- jroot@uchicago.edu}\, and Evan Sadler\thanks{Boston University -- edsadler@bu.edu}}
\date{\today}

\begin{document}

\maketitle

\begin{abstract}
We provide interpretable axiomatic foundations for utilities used in network games and identify several principled generalizations.  First, we demonstrate that a ubiquitous feature of network games, bilateral strategic interactions, is equivalent to having player utilities that are additively separable across opponents.  Common utilities based on a linear aggregate of opponent actions are strategically equivalent to additively separable utilities.  Moreover, assuming real-valued actions, we show that a constant rate of substitution between opponents implies a utility that is \emph{linear} in opponent actions.  Finally, we identify precise conditions---linear best replies and midpoint indifference---that pin down the classic linear-quadratic utility.


\end{abstract}
\newpage
\section{Introduction}

\hspace{1 pc}
For tractability, network games typically employ simple utility functions.  Most common is the linear-quadratic form that \citet{Ballesteretal2006} popularized:  we have actions $S_i = \bb{R}_+$ for each player $i$, and
\begin{equation}\label{eq:linquad}
u_i(\mathbf{s}) = b_i s_i + s_i \sum_{j \neq i} g_{ij} s_j - \frac{s_i^2}{2}
\end{equation}

\noindent for constants $b_i$ and $\{g_{ij}\}_{j \neq i}$.  Using this, researchers study both strategic complements and substitutes (taking $g_{ij}$ positive or negative, respectively) and evaluate targeted interventions. Analytically, these preferences are especially convenient as best responses are linear in opponents' actions, so equilibrium computation reduces to solving a linear system. Nevertheless, we lack a firm grasp on what substantive restrictions these preferences entail.  What are we assuming, and how might we judge whether real preferences are consistent with \eqref{eq:linquad}?

We provide interpretable axioms under which preferences are represented by utilities of this form, and we identify several other canonical utilities that relax key assumptions.  We decompose the problem into two steps. First, we establish conditions under which preferences are separable over opponent actions, admitting a utility of the form
\begin{equation}\label{eq:utility1}
    u_i(\mathbf{s}) = \sum_{j \neq i} g_{ij}(s_i, s_j).
\end{equation}
The key axioms formalize the idea that strategic interactions are bilateral. Roughly speaking, bilateral strategic interactions means that how $j$'s action influences player $i$ is independent of what some other opponent $k$ does. For this step, we need not assume that actions are real-valued.  In a second step, we assume separability and real-valued actions, and we show that three further axioms pin down \eqref{eq:linquad}. A constant rate of substitution across opponents yields utilities that are linear in opponents' actions:
    \begin{equation}\label{eq:linearopp}
    u_i(\mathbf{s}) = b(s_i) + \gamma(s_i) \sum_{j \neq i} g_{ij} s_j.
    \end{equation}
\noindent Linear best replies and a novel axiom, \emph{midpoint indifference}, then uniquely axiomatize \eqref{eq:linquad}. 

 As a starting point for our analysis, we distinguish two senses in which strategic interactions can be bilateral.  Player $i$'s preferences satisfy \emph{opponent independence} if whenever $u_i(s_i, s_j, s'_{-ij}) \geq u_i(s_i, s_j, s_{-ij})$ for some action $s_{j}$, then also $u_i(s_i, s'_j, s'_{-ij}) \geq u_i(s_i, s'_j, s_{-ij})$ for any other action $s'_{j}$.  Player $i$ has a fixed preference ranking over the actions of her other opponents, no matter what $j$ does.  Her preferences satisfy \emph{strategic independence} if we can never have the following four comparisons with at least one strict:
$$u_i(s_i, s_j, s_{-ij}) \geq u_i(s'_i, s_j, s_{-ij}), \quad u_i(s'_i, s'_j, s_{-ij}) \geq u_i(s_i, s'_j, s_{-ij}),$$
$$u_i(s'_i, s_j, s'_{-ij}) \geq u_i(s_i, s_j, s'_{-ij}), \quad u_i(s_i, s'_j, s'_{-ij}) \geq u_i(s'_i, s'_j, s'_{-ij}).$$

\noindent Intuitively, this means that whenever $j$ changes his action from $s_j$ to $s'_j$, this can only shift $i$'s preferences between $s_i$ and $s'_i$ in one direction.  The first two comparisons say that, if other opponents play $s_{-ij}$, then player $j$ switching from $s_j$ to $s'_j$ causes $i$'s preference to flip in favor of $s'_i$.  Strategic independence implies that we cannot then find some other profile $s'_{-ij}$ at which $j$ switching from $s_j$ to $s'_j$ causes $i$'s preference to flip in favor of $s_i$.

These natural properties correspond to utilities that are separable across opponents.  Under mild technical restrictions, Theorem \ref{theo:main1} establishes that if player $i$'s preferences jointly satisfy opponent and strategic independence, then we can represent these preferences via a separable utility as in (\ref{eq:utility1}).  In this case, we say that player $i$ has \emph{separable preferences}.  Moreover, this representation is \emph{cardinal}:  it is unique up to a positive affine transformation. 

While the linear-quadratic utility \eqref{eq:linquad} is clearly separable, common generalizations satisfy strategic independence but not joint independence.\footnote{See, for instance, \citet{Belhajetal2014}, \citet{PariseOzdaglar2019}, and \citet{ZenouZhou2023} who adopt utility functions that depend on linear aggregates of opponent actions.}  Accordingly, such preferences are not separable. However, if we weaken our requirements so that \eqref{eq:utility1} only represents player $i$'s preferences over her \emph{own actions}, an additively separable utility can describe her strategic incentives.  Under stronger technical restrictions, Theorem \ref{theo:main2} shows that, if player $i$'s preferences satisfy strategic independence alone, then a utility of the form \eqref{eq:utility1} represents her preferences over $S_i$ for each opponent profile $s_{-i}$.  As before, the representation is cardinal, and we say that player $i$ has \emph{strategically separable preferences}.  A corollary highlights that preferences based on aggregates of neighbors' actions are strategically separable.

An alternative way to understand these results is that, whenever player $i$'s preferences preclude an additively separable representation, we can always find a small set of comparisons, no more than four, that by themselves demonstrate that fact.  This view illuminates why we need further conditions in Theorem \ref{theo:main2}.  Strategic independence restricts $i$'s preferences over $S_i$, while saying little about her preferences over all of $S$.  Relative to Theorem \ref{theo:main1}, we have far fewer comparisons we can use to show that $i$'s preferences are inconsistent with \eqref{eq:utility1}.  To compensate, we require that opponent action sets are sufficiently rich.  We discuss our conditions in detail in Section \ref{sec:assumptions}.  Both findings build upon an older literature that studies additively separable utilities in consumer theory.  In our proofs, we first apply a Theorem from \citet{Wakker1989} to obtain an additively separable utility for each fixed $s_i \in S_i$.  Our principal innovation is to engineer, from this family of utilities, a single utility over the entire space of strategy profiles, stitching together the utilities for each fixed $s_i$.  





With these results in hand, we turn to axioms that deliver familiar functional forms.  Given real-valued actions, we say that preferences feature a \emph{constant rate of substitution} between opponents $j$ and $k$ if, for some constant $\delta$, increasing $s_j$ by $1$ always has the same effect on $i$'s preferences over $S_i$ as does increasing $s_k$ by $\delta$.  If $i$'s preferences feature a constant rate of substitution between any two opponents, we say that she has \emph{CRS preferences}.  This class contains the linear-quadratic form \eqref{eq:linquad} as well as generalizations like
\begin{equation}\label{eq:linaggutility}
u_i(\mathbf{s}) = b(s_i) + \gamma(s_i) \psi\left(\sum_{j \neq i} g_{ij} s_j\right).
\end{equation}

\noindent As long as the function $\psi$ in \eqref{eq:linaggutility} is strictly monotone, such preferences are strategically separable, but Theorem \ref{theo:CRS} asserts an even stronger conclusion.  We can always represent strategically separable CRS preferences via a utility of the form
\begin{equation}\label{eq:CRSutility}
\tilde{u}_i(\mathbf{s}) = \tilde{b}(s_i) + \tilde{\gamma}(s_i) \sum_{j \neq i} g_{ij} s_j
\end{equation}

\noindent for some functions $\tilde{b}$ and $\tilde{\gamma}$.

As one example, suppose actions are real-valued and
$$u_i(\mathbf{s}) = b_i s_i - s_i^2 + s_i \ln \left(\sum_{j \in G_i} g_{ij} s_j\right).$$

\noindent Higher opponent actions clearly favor higher $s_i$, and for $s'_i > s_i$, one can verify that player $i$ prefers $s'_i$ exactly when 
$\sum_{j \in G_i} g_{ij} s_j > e^{s'_i + s_i - b_i}$.  If instead we take the utility
$$\tilde{u}_i(\mathbf{s}) = -e^{s_i} - e^{b_i - s_i}\sum_{j \neq i} g_{ij}s_j,$$
\noindent which takes the form \eqref{eq:CRSutility}, we can compute the difference $\tilde{u}_i(s'_i, s_{-i}) - \tilde{u}_i(s_i, s_{-i})$ to find again that player $i$ prefers $s'_i$ exactly when $\sum_{j \in G_i} g_{ij} s_j > e^{s'_i + s_i - b_i}$.  While the utility $\tilde{u}_i$ captures the same strategic incentives as $u_i$, our example also illustrates the limits of this equivalence.  According to $u_i$, player $i$ benefits when opponents take higher actions, but according to $\tilde{u}_i$, player $i$ prefers that opponents take lower actions.  Our result thus justifies restricting attention to utilities \eqref{eq:CRSutility} for finding equilibria and studying comparative statics, but welfare analysis requires more caution.


With two further axioms, one familiar and one novel, we arrive at the classic linear-quadratic utility \eqref{eq:linquad}.  Linear best replies are key to the tractability of linear-quadratic payoffs, but assuming linear best replies is not enough on its own because doing so imposes no restrictions on preferences away from the optimal action.  To get payoffs of the form \eqref{eq:linquad} we additionally need \emph{midpoint indifference}:  if $s_i$ is a best response to $s_{-i}$ and $s'_i$ is a best response to $s'_{-i}$, then player $i$ is indifferent between $s_i$ and $s_i'$ at $\frac{s_{-i}+s_{-i}'}{2}$.  These two conditions, on top of a constant rate of substitution, imply that a player's preferences are strategically equivalent to a linear-quadratic utility \eqref{eq:linquad}.

In a final section, we do away with technical restrictions and characterize precisely when a player's preferences are strategically separable.  We show that preferences are strategically separable if, and only if, they do not admit a \emph{balanced sequence}.  Take any two sequences of strategy profiles $\{r^1, r^2,...,r^m\}$ and $\{t^1, t^2,...,t^m\}$ such that $r^k$ and $t^k$ differ only in player $i$'s action, and for each opponent $j$, each pair of actions $(s_i, s_j)$ appears the same number of times in both sequences.  If we can represent $i$'s preferences with a utility of the form \eqref{eq:utility1}, then we must have
$$\sum_{k=1}^m u_i(r^k) = \sum_{k=1}^m u_i(t^k).$$

\noindent Player $i$ therefore cannot prefer $r^k$ to $t^k$ for each pair in the sequence, unless she is always indifferent.  A balanced sequence is precisely a sequence of profile pairs that violates this requirement.  In the main text, we use a variant of Farkas' lemma to prove this result for finite games.  In an appendix, we use logical compactness to extend this to the general case.  Relating to our earlier findings, note that strategic independence rules out particular types of balanced sequences of length $4$.  We show that in the general case, one may, in fact, need to check arbitrarily long sequences to ensure that we can find a separable utility representation.  This last result thus sheds light on what our technical assumptions buy us: checking short sequences is enough.

Why should an economist care about these findings?  First, network games offer a rare example of a rich family that admits tractable equilibrium analysis, which has made them a popular tool for applied modeling, and our results provide interpretable foundations for utilities used in this literature. Second, there is growing consensus that the linear-quadratic framework is too restrictive in many applications, and our results point towards principled generalizations. Third, we provide revealed-preference tests for whether preferences are consistent with these assumptions. Finally, our findings offer a template for examining other natural preference restrictions. For instance, certain types of symmetry, or assumptions on rates of substitution, can help identify other important classes of utility functions.




\subsection*{Related Work}

\hspace{1 pc}
Most studies of network games, including some of the most recent, build directly on the linear-quadratic framework of \citet{Ballesteretal2006}.  This includes several theoretical analyses of public goods \citep{BramoulleKranton2007, Bramoulleetal2014, Gerkeetal2024, BervoetsMakihara2025}, status concerns \citep{Immorlicaetal2017,Langtry2023}, multiple activities \citep{Chenetal2018}, and targeted incentives \citep{Galeottietal2020, Dasarathaetal2025,ClaveriaMayoletal2025}, as well as empirical studies of academic peer effects \citep{CalvoArmengoletal2009, Conleyetal2024} and crime \citep{Diazetal2021}.\footnote{Beyond these references, there is a very large empirical literature on peer effects based on a linear-in-means specification.  The most natural way to microfound these econometric models is through linear-quadratic preferences.}

At the same time, there is increasing interest in moving beyond this restrictive functional form.  Indeed, recent empirical work suggests that estimates based on linear specifications can lead to significant policy errors when peer effects are non-linear \citep{Boucheretal2024}.  So far, theoretical work has taken a simple, pragmatic approach, making incremental adjustments to popular utilities.  In essentially all of these models, preferences depend on a linear aggregate of neighbors' actions:  we have
\begin{equation*}
u_i(\mathbf{s}) = b_i s_i + s_i \psi\left(\sum_{j \neq i} g_{ij} s_j\right) - \frac{s_i^2}{2},
\end{equation*}

\noindent for an increasing function $\psi$ \citep[See, for instance,][]{Belhajetal2014,PariseOzdaglar2019,PariseOzdaglar2023, ZenouZhou2023}.  With a linear $\psi$, this includes the classic model as a special case.  There is also growing interest in endogenizing the network by letting agents choose both their action and the extent of their bilateral links.\footnote{ See, for instance, \citet{sadler2021games} and \cite{cerreia2023nonlinear}.} If agents can choose both the $g_{ij}$ and the $s_i$, the utility representation from \citet{Ballesteretal2006} remains additively separable in the joint action $(s_i,g_i)$. In this way, additively separable preferences can encode richer strategic information than in the classical setting \eqref{eq:linquad}.

We take a fundamentally different approach to generalizing the utilities used in network games.  We highlight key properties of preferences, opponent and strategic independence, and identify a canonical class of utilities that satisfy them.  This makes clear precisely what we assume when we write down a utility of this form.  Moreover, it ensures results derived using this quasi-functional form apply to any preferences that satisfy our independence properties.  Our exercise is close in spirit to \citet{Sadler2022}, which uses a robust notion of centrality to identify players who take higher equilibrium actions in any (symmetric) network game of strategic complements.  

The vast literature on network games includes several models that do not fit within our framework.  In particular, coordination games with discrete action sets \citep[e.g.][]{Blume1993,Morris2000,Reich2023,JacksonStorms2023} necessarily fail to satisfy our richness conditions and therefore fall outside our scope.  While we do present a version of our main result for games with finite action sets in Section \ref{sec:cyclic}, checking the required conditions on preferences is much more onerous.

While our results are most obviously related to the current literature on network games, our techniques draw more heavily from a much older literature on additively separable utility representations.  In the context of consumer theory, \citet{Debreu1959} shows that, with at least three goods available, a condition analogous to opponent independence implies that we can represent a consumer's preferences over bundles $\mathbf{x}$ as
$$u(\mathbf{x}) = \sum_{i = 1}^n v_i(x_i).$$

\noindent Moreover, this utility representation is cardinal.  Versions of this result, with slightly weaker technical assumptions and more direct proofs, appear in \citet{Fishburn1970} and \citet{Wakker1989}---we rely on the last of these in our own arguments.  One can similarly derive cardinal utility representations from orders on \emph{utility differences} \citep[e.g.][]{Kobberling2006}, a method that appears in our proofs.



\section{Utility Representations in Network Games}

\hspace{1 pc}
Recall a normal form game $\Gamma = (N, \{S_i\}_{i \in N}, \{u_i\}_{i \in N})$ consists of a set of players $N$, a set of actions $S_i$ for each player $i \in N$, and a utility $u_i \, : \, \prod_{i \in N} S_i \to \bb{R}$ mapping each profile of actions to a real number. We write $S = \prod_{i \in N} S_i$ for the set of action profiles and $s$ for a generic element of $S$.  Player $i$'s utility $u_i$ represents an underlying preference order $\succeq_i$ on $S$.  Two utility functions $u_i$ and $v_i$ are \textbf{equivalent} if they represent the same order $\succeq_i$.  For each profile $s_{-i}$ of opponent actions, we can restrict $\succeq_i$ to the space $S_i \times \{s_{-i}\}$, to obtain an order $\succeq_{i, s_{-i}}$ that represents $i$'s preferences over $S_i$ when opponents play $s_{-i}$.  The family of orders $\{\succeq_{i, s_{-i}}\}_{s_{-i} \in S_{-i}}$ fully captures $i$'s preferences over her own action choice and are therefore sufficient to describe the strategic properties of $\Gamma$.\footnote{One might object that this family of orders is insufficient to capture preferences over \emph{mixtures} of actions, but at this point, action sets are completely general.  To consider mixed strategies, one could take $S_i$ to be a set of probability distributions.  However, in general, the utility we derive is different from the Von Neumann-Morgenstern utility.  In Appendix \ref{appendix: mixed strategies}, we identify conditions under which they are the same.} Two utility functions $u_i$ and $v_i$ are \textbf{strategically equivalent} if they induce the same family of orders over $S_i$.\footnote{Equivalence is stronger than strategic equivalence.  A utility function $u_i(\mathbf{s})$ is strategically equivalent to $u_i(\mathbf{s}) + f(s_{-i})$ for any function $f(s_{-i})$, but our choice of $f$ can easily change the corresponding preference order on $S$.  Two utilities $u_i$ and $v_i$ are equivalent if and only if $u_i(\mathbf{s}) = f(v_i(\mathbf{s}))$ for some strictly increasing function $f$.  They are strategically equivalent if and only if $u_i(\mathbf{s}) = f(v_i(\mathbf{s}), s_{-i})$ for some function $f$ that is strictly increasing in its first argument.}

We are interested in conditions under which a player's utility $u_i$ is equivalent, or strategically equivalent, to utilities $v_i$ in the following special classes:

\begin{enumerate}
    \item \textbf{Additively separable across opponents:}
    \begin{equation}\label{eq:addseputility}
    v_i(\mathbf{s}) = \sum_{j \neq i} g_{ij}(s_i, s_j)
    \end{equation}

    for some functions $\{g_{ij}\}_{j \neq i}$.

    \item \textbf{Linear in opponent actions}, with real-valued actions:
    \begin{equation}\label{eq:linearopp}
    v_i(\mathbf{s}) = b(s_i) + \gamma(s_i) \sum_{j \neq i} g_{ij} s_j
    \end{equation}

    for some real constants $\{g_{ij}\}_{j \neq i}$ and functions $b(s_i)$ and $\gamma(s_i)$.

    \item \textbf{Linear-quadratic}, with real-valued actions:
    \begin{equation}\label{eq:linearquad}
    v_i(\mathbf{s}) = b s_i + s_i \sum_{j \neq i} g_{ij} s_j - \frac{s_i^2}{2}
    \end{equation}

    for some real constants $\{g_{ij}\}_{j \neq i}$ and $b$.
    
\end{enumerate}

\noindent We say that player $i$'s preferences are \textbf{separable} if $u_i$ is equivalent to some $v_i$ of the form \eqref{eq:addseputility}.  Player $i$'s preferences are \textbf{strategically separable} if $u_i$ is strategically equivalent to $v_i$ of the form \eqref{eq:addseputility}.\footnote{When preferences are separable, the corresponding utility functions are unique up to positive affine transformations, but the specific terms $g_{ij}$ in the sum \eqref{eq:addseputility} are not.  Given a utility of this form, and any function $h(s_i)$, we can pick any two opponents $j$ and $k$ and substitute $\tilde{g}_{ij}(s_i, s_j) = g_{ij}(s_i, s_j) + h(s_i)$ and $\tilde{g}_{ik}(s_i, s_k) = g_{ik}(s_i, s_k) - h(s_i)$, giving the same utility function.}  Utilities that are linear in opponent actions are clearly separable and have further structure.  Linear-quadratic utilities impose even stronger assumptions on preferences.  Our main results give precise conditions under which a player's preferences are equivalent, or strategically equivalent, to utilities of these forms.

Taking separability or strategic separability as given, we can readily state the additional conditions needed for a utility that is linear in opponent actions or linear-quadratic.  Assuming actions are real-valued, the key condition required for the representation \eqref{eq:linearopp} is that preferences exhibit a \textbf{constant rate of substitution}.

\begin{definition}\label{def:CRS}
With real-valued actions, player $i$'s preferences exhibit a \textbf{constant rate of substitution} between opponents $j$ and $k$ if there is a constant $\delta$ such that $\succeq_{i, s_{-i}} = \succeq_{i, s'_{-i}}$ whenever $s_\ell = s'_\ell$ for $\ell \notin \{j,k\}$ and $\delta(s'_j - s_j) = s_k - s'_k$.  If player $i$'s preferences feature a constant rate of substitution between any two opponents whose actions can change $i$'s preference ranking, we say that player $i$ has \textbf{CRS preferences}.
\end{definition}

A constant rate of substitution between $j$ and $k$ means that increasing $j$'s action by $1$ has the same effect on $i$'s preferences as increasing $k$'s action by $\delta$.  Payoffs that depend on a linear aggregate $\sum_{j \neq i} g_{ij}s_j$ clearly exhibit CRS preferences:  whenever $g_{ik} \neq 0$, we can take $\delta = \frac{g_{ij}}{g_{ik}}$ in Definition \ref{def:CRS}.  If CRS preferences are strategically separable, then the corresponding additively separable utility is linear in opponent actions, taking the form \eqref{eq:linearopp}.\footnote{Our later results on additively separable representations imply that the representation \eqref{eq:linearopp} is unique up to positive affine transformations, so if all players in a network game have strategically separable CRS preferences, this means that the corresponding weighted graph $G = \{g_{ij}\}_{i,j \in N}$ is unique up to left multiplication by a positive diagonal matrix.}

\begin{Theorem}\label{theo:CRS}
Suppose actions are real-valued, and player $i$ has continuous and strategically separable CRS preferences.  Player $i$'s preferences are strategically equivalent to a utility
$$u_i(\mathbf{s}) = b(s_i) + \gamma(s_i) \sum_{j \neq i} g_{ij} s_j.$$

\end{Theorem}

\begin{proof}
    We provide a full proof in Section \ref{sec:agg}.
\end{proof}

To get the linear-quadratic utility \eqref{eq:linearquad}, we need two further conditions.  Linear-quadratic utilities famously lead to linear best-response maps, but assuming linear best replies alone is, perhaps surprisingly, not enough.  Suppose $S_i = \bb{R}$, and utilities take the form
$$u_i(\mathbf{s}) = (b_i - s_i) e^{s_i} + e^{s_i} \sum_{j \neq i} g_{ij} s_j.$$

\noindent The first order condition reduces to
$$(b_i - 1 - s_i) e^{s_i} + e^{s_i} \sum_{j \neq i} g_{ij} s_j = 0,$$

\noindent which has the unique solution
$$s_i = b_i - 1 + \sum_{j \neq i} g_{ij} s_j,$$

\noindent and one can readily verify that the second-order condition for a maximum holds.  The missing condition is \textbf{midpoint indifference}.

\begin{definition}\label{def:midpointind}
With real-valued actions, player $i$'s preferences satisfy \textbf{midpoint indifference} if, whenever $s_i$ is a best response to $s_{-i}$ and $s'_i$ is a best response to $s'_{-i}$, player $i$ is indifference between $s_i$ and $s'_i$ against $\frac{s_{-i} + s'_{-i}}{2}$.
\end{definition}

Given strategic separability, the classic linear-quadratic utility is equivalent to a constant rate of substitution between neighbors, linear best replies, and midpoint indifference.

\begin{Theorem}\label{theo:linquad}
Suppose actions are real-valued, and player $i$ has continuous and strategically separable CRS preferences.  If additionally player $i$ has a linear best response, and satisfies midpoint indifference, then player $i$'s preferences are strategically equivalent to a linear-quadratic utility
$$u_i(\mathbf{s}) = b s_i + s_i \sum_{j \neq i} g_{ij} s_j - \frac{s_i^2}{2}.$$
\end{Theorem}

\begin{proof}
    We provide a full proof in Section \ref{sec:agg}.
\end{proof}

Note that Theorems \ref{theo:CRS} and \ref{theo:linquad} are written assuming the weaker strategic separability, and the conclusions correspondingly deliver strategic equivalence.  Using identical arguments, we could write versions of these theorems assuming separability and delivering equivalence, but the utility representations are slightly different.  For CRS preferences, the corresponding result would give a utility
$$u_i(\mathbf{s}) = b(s_i) + \gamma(s_i) \sum_{j \neq i} g_{ij} s_j + \sum_{j \neq i} f_j(s_j),$$

\noindent adding an extra set of terms that depend only on opponent actions.  Similarly, given CRS preferences, linear best replies, and midpoint indifference, the representation would become
$$u_i(\mathbf{s}) = b s_i + s_i \sum_{j \neq i} g_{ij} s_j - \frac{s_i^2}{2} + \sum_{j \neq i} f_j(s_j).$$

\noindent The extra terms in each case affect $i$'s preferences over opponent profiles without impacting $i$'s preferences over her own actions $S_i$.

Given these findings, the principal technical hurdle is to uncover when preferences are separable or strategically separable. In the next section, we introduce independence axioms, opponent and strategic independence, that characterize separable and strategically separable preferences.  We develop our main technical results in Sections \ref{sec:separable} and \ref{sec:stratseparable}.  The first result shows that, under mild technical restrictions, the preference order $\succeq_i$ on $S$ is separable if and only if it jointly satisfies opponent and strategic independence.  The second result shows that, under stronger conditions, the family of orders $\{\succeq_{i,s_{-i}}\}_{s_{-i} \in S_{-i}}$ on $S_i$ is strategically separable if and only if it satisfies strategic independence.  In Section \ref{sec:cyclic}, we discuss the importance of our technical assumptions, and we provide an alternative characterization of strategic separability without them.


\section{Independence Axioms}

\hspace{1 pc}
In the introduction, we noted a ubiquitous feature of preferences in network games:  strategic interactions are bilateral.  We formalize this idea through two properties, \emph{opponent independence} and \emph{strategic independence}.  Going forward, we take $N = \{0,1,2,...,n\}$ for the set of players and focus on representing player $0$'s preferences.  We write $s_0$ for a generic action of this focal player and $s_+$ for a generic profile of actions for her $n$ opponents.  We use $s_{-i}$ to denote a profile of actions for \emph{opponents} other than $i$, and $s_A$ to denote a profile of actions for the subset of opponents $A$.

\begin{definition}\label{def:independence}
    A preference order $\succeq_0$ on $S$ satisfies \textbf{opponent independence} if
    $$(s_0, s_i, s'_{-i}) \succeq_0 (s_0, s_i, s_{-i}) \quad \implies \quad (s_0, s'_i, s'_{-i}) \succeq_0 (s_0, s'_i, s_{-i})$$

    \noindent for any opponent $i$, any $s_i, s'_i \in S_i$, and any $s_{-i}, s'_{-i} \in S_{-i}$.  The corresponding family of orders $\{\succeq_{s_+}\}_{s_+ \in S_+}$ satisfies \textbf{strategic independence} if for any actions $s_0, s'_0 \in S_0$, any opponent $i > 0$, and any opponent action profiles $s_+, s'_+ \in S_+$, we have
    $$s_0 \succeq_{s_+} s'_0, \quad s'_0 \succeq_{s'_i, s_{-i}} s_0, \quad s'_0 \succeq_{s_i, s'_{-i}} s_0 \quad \implies \quad s'_0 \succeq_{s'_+} s_0,$$

    \noindent with the last comparison being strict if any one of the first three is strict.  
\end{definition}

Interpreting opponent independence is straightforward:  our focal player's preferences over action profiles for opponents other than $i$ cannot depend on what opponent $i$ is doing.  This is clearly necessary for an additively separable representation \eqref{eq:addseputility}, and it is identical to ``coordinate independence'' as used in consumer theory to characterize preferences over consumption bundles that are separable across goods.  Since opponent independence restricts preferences \emph{over opponent actions}, it only appears in our characterization of separable preferences where we start from an order $\succeq_0$ over all of $S$.\footnote{The condition $(s_0, s_i', s_{-i}) \succeq_0 (s_0, s_i, s_{-i})  \implies (s_0, s'_i, s'_{-i}) \succeq_0 (s_0, s_i, s_{-i}')$ is not sufficient to generate an additively separable utility representation over $s_+$ for each $s_0$.}

Strategic independence is new to our paper, and it concerns how changes in opponent actions affect our focal player's preferences over her own actions.  The essence of strategic independence is that if a given opponent $i$'s switching from $s_i$ to $s'_i$ can make $s'_0$ more attractive relative to $s_0$, then it must always make $s'_0$ relatively more attractive than $s_0$, no matter what other opponents are doing.  The first two comparisons assert that, when other opponents play $s_{-i}$, switching from $s_i$ to $s'_i$ flips our focal player's preference from $s_0$ to $s'_0$.  If preferences satisfy strategic independence, then there cannot be another profile $s'_{-i}$ such that the switch from $s_i$ to $s'_i$ flips our focal player's preference in the opposite direction.



Our result on separable preferences requires a notion of joint independence that is somewhat stronger than assuming both opponent and strategic independence.\footnote{We find it helpful to draw a parallel with differentiability of multivariate functions.  Having partial derivatives in each variable does not guarantee the function is differentiable overall because it could still be poorly behaved if we move along a diagonal.  Our definition of joint independence analogously contemplates simultaneous changes in own and opponents' actions.}

\begin{definition}\label{def:jointind}
    The order $\succeq_0$ is \textbf{jointly independent} if, for any actions $s_0, s'_0 \in S_0$, any two opponents $i,j > 0$, and any opponent action profiles $s_+, s'_+, t_+, t'_+ \in S_+$, we have
    $$(s_0, s_+) \succeq_0 (s'_0, s'_+), \quad (s'_0, t'_{j}, s'_{-j}) \succeq_0 (s_0, t_i, s_{-i}), \quad (s'_0, s'_{j}, t'_{-j}) \succeq_0 (s_0, s_i, t_{-i})$$
    $$\implies \quad (s'_0, t'_+) \succeq_0 (s_0, t_+),$$

    \noindent with the last comparison being strict if any one of the first three is strict.
\end{definition}

Joint independence implies both opponent independence and strategic independence as special cases. Opponent independence is the special case in which $s_0 = s'_0$, $i = j$, $s'_i = s_i$, $t'_i = t_i$, and $s_{-i} = s'_{-i}$ while strategic independence is the special case in which $t_+=t_+'$, $s_+=s_+'$ and $i=j$. Joint independence imposes additional restrictions if we change both the focal player's action and opponent actions simultaneously.




\section{Separable Preferences}\label{sec:separable}

\hspace{1 pc}
Given an order $\succeq_0$ on a space of strategy profiles $S$, when can we represent it using an additively separable utility \eqref{eq:addseputility}?  Under mild technical restrictions, the answer is precisely when $\succeq_0$ satisfies joint independence.  We divide our conditions into two formal assumptions.  The first is indispensable while the second simplifies our argument---at the end of the section, we discuss how to relax it.  Fixing an action $s_0$, we say that opponent $i$ is \textbf{essential at} $s_0$ if $(s_0, s'_i, s_{-i}) \succ_0 (s_0, s_i, s_{-i})$ for some $s_i, s'_i \in S_i$ and $s_{-i} \in S_{-i}$.


\begin{assumption}\label{as:separable1}
    Each action set $S_i$ is a connected topological space, the preference order $\succeq_0$ on $S$ is continuous.  There exists $s_0^*$ with at least two essential opponents.
\end{assumption}



The most meaningful restriction here is that preferences are continuous.  Note that if the last part of Assumption \ref{as:separable1} fails, then the utility is trivially separable across opponents because our focal player cares about at most one opponent.  The reason we exclude this case is because without two essential opponents, the utility representation is no longer cardinal---any monotone transformation remains separable across opponents.  Our second assumption imposes an additional richness condition.


\begin{assumption}\label{as:separable2}
    For some $s_0^*$ with at least two essential opponents, we have that for any profile $(s_0,s_+)$ we can find $s^*_+ \in S_+$ such that $(s_0,s_+) \simeq_0 (s_0^*, s^*_+)$.
\end{assumption}

Assumption \ref{as:separable2} says that there exists an $s_0^*$ such that, for any profile $s$, we can find a profile of opponent actions $s_+^*$ to make our focal player indifferent between $s$ and $(s_0^*,s_+^*)$.  While not strictly necessary for our conclusion, this assumption allows a more transparent argument, and we discuss at the end of this section how to extend beyond this case.  The standard linear-quadratic utility
$$u(s) = b s_0  + s_0 \sum_{i = 1}^n \gamma_{i} s_i- \frac{s_0^2}{2}$$

\noindent satisfies this assumption as long as either i) the action of at least one essential opponent can range over all of $\bb{R}$, or ii) actions are unbounded (e.g., in the usual case with $S_i = \bb{R}_+$), and at least one weight $\gamma_i$ is positive and another is negative.  To preview how we can dispense with this assumption, if all weights have the same sign, and $S_i = \bb{R}_+$ for each opponent $i$, then assumption \ref{as:separable2} fails. However, if we restrict the focal player to actions $S_0 = \left[\frac{1}{N}, N\right]$ for any positive integer $N$, the assumption is restored, and we can take a limit to cover all of $\bb{R}_+$.

\begin{Theorem}\label{theo:main1}
    Suppose Assumptions \ref{as:separable1} and \ref{as:separable2} hold for action sets $\{S_i\}_{i=0}^n$ and the preference order $\succeq_0$.  The order $\succeq_0$ is separable if and only if it satisfies joint independence.  Moreover, the representation of $\succeq_0$ via a utility of the form
    $$u(s) = \sum_{i = 1}^n g_i(s_0, s_i)$$

    \noindent is unique up to a positive affine transformation.
\end{Theorem}

We present the full proof of Theorem \ref{theo:main1} in Section \ref{sec:proof1}.  The first step is a straightforward application of an old result in consumer theory.  Fixing an action $s_0$ for our focal player, earlier work gives us an additively separable utility that represents preferences over opponent profiles $S_+$.  We state the relevant result as a lemma, and provide references, in Section \ref{sec:prelim}.  The novel part of our proof shows that, using joint independence, we can merge the utilities we obtain for each fixed $s_0$ into a single utility that covers all of $S_0$.

\subsection{Proof Preliminaries}\label{sec:prelim}

\hspace{1 pc}
Suppose $X = X_1 \times X_2 \times ... \times X_n$---classically, we would have in mind the space of bundles involving $n$ different commodities.  A binary relation $\succeq$ on $X$ is \textbf{coordinate independent} if for any $i$ we have
$$(x_i, x_{-i}) \succeq (x_i, y_{-i}) \quad \iff \quad (y_i, x_{-i}) \succeq (y_i, y_{-i})$$


\noindent for any two $x_i, y_i \in X_i$.  Coordinate $i$ is \textbf{essential} if $(x_i, x_{-i}) \succ (y_i, x_{-i})$ for some $x_i, y_i, x_{-i}$.  Various authors going back to \citet{Debreu1959} have identified conditions under which a coordinate independent preference order over bundles admits an additively separable utility representation:
$$u(\mathbf{x}) = \sum_{i=1}^n v_i(x_i).$$

\noindent Technical conditions on the sets $\{X_i\}$ vary, but most results of this sort require a continuous preference order and at least three essential coordinates.  We rely on a version from \citet{Wakker1989}, which requires that each $X_i$ is a connected topological space and $\succeq$ is a continuous weak order on $X$.

With only two essential coordinates, we need a further condition.  The weakest one of which we are aware is the \textbf{hexagon condition} that \citet{Wakker1989} identifies.  The binary relation $\succeq$ on $X = X_1 \times X_2$ satisfies the hexagon condition if
$$(x_1, x_2) \simeq (y_1, y_2) \quad \text{and} \quad (z_1,y_2) \simeq (y_1, x_2) \simeq (x_1, z_2) \quad \implies \quad (y_1, z_2) \simeq (z_1, x_2).$$

\noindent  Given an additively separable utility, the hypotheses in the hexagon condition state that
$$v_1(x_1) + v_2(x_2) = v_1(y_1) + v_2(y_2), \quad \text{and}$$
$$v_1(z_1) + v_2(y_2) = v_1(y_1) + v_2(x_2) = v_1(x_1) + v_2(z_2)$$

\noindent Adding the first and last two quantities in the second line respectively gives
$$v_1(z_1) + v_2(x_2) + v_1(y_1) + v_2(y_2) = v_1(y_1) + v_2(z_2) + v_1(x_1) + v_2(x_2).$$

\noindent Subtracting $v_1(y_1) + v_2(y_2)$ from the first, and $v_1(x_1) + v_2(x_2)$ from the second, shows that
$$v_1(z_1) + v_2(x_2) = v_1(y_1) + v_2(z_2).$$

\noindent Hence, the hexagon condition is necessary to have an additively separable representation.  With three essential coordinates, one can derive it from coordinate independence, but with only two, we must assume it separately.  Note that in our model, opponent independence is exactly coordinate independence on $\{s_0\} \times S_+$ for each action $s_0 \in S_0$, and joint independence implies the hexagon condition on $\{s_0\} \times S_+$ for each action $s_0 \in S_0$.

For ease of reference, we state the full result from \citet{Wakker1989} here.

\begin{lemma}[Wakker, 1989, Theorem III.4.1]\label{lem:wakker}
    Suppose $X = X_1 \times X_2 \times ... \times X_n$, with $n \geq 3$, is a connected topological space, and $\succeq$ is a continuous weak order on $X$.  If $\succeq$ is coordinate independent, and at least three coordinates are essential, then there exist real-valued functions $\{v_i\}_{i = 1}^n$ such that $\mathbf{x} \succeq \mathbf{y}$ if and only if
    $$\sum_{i = 1}^n v_i(x_i) \geq \sum_{i=1}^n v_i(y_i).$$

    \noindent Moreover, the functions $\{v_i\}_{i = 1}^n$ are unique up to a positive affine transformation.  In particular, if $\{\tilde{v}_i\}_{i=1}^n$ is another such representation, then
    $$\tilde{v}_i(x_i) = \lambda v_i(x_i) + c_i, \quad i=1,2,...,n$$

    \noindent for some $\{c_i\}_{i = 1}^n$ and $\lambda > 0$.  With two essential coordinates, the same conclusion holds if $\succeq$ additionally satisfies the hexagon condition.
\end{lemma}


We invoke Lemma \ref{lem:wakker} twice in our proof.  The first time, we apply it to obtain an additively separable utility on $\{s_0\} \times S_+$.  We then use these utilities to construct a \emph{difference order} on all of $S$.  A difference order is an order on \emph{pairs} of strategy profiles, with the interpretation that $st \succeq s't'$ if
$$u(s) - u(t) \geq u(s') - u(t').$$

\noindent It is well-known that a difference order satisfying natural regularity conditions yields a \emph{cardinal} utility representation:  if two utility functions correspond to the same difference order, then one is a positive affine transformation of the other \citep[See, for instance, ][]{Kobberling2006}.

In fact, the result for difference orders is a special case of Lemma \ref{lem:wakker} with $n=2$.  Taking $X_1 = X_2$, if we have
$$(x, x) \simeq (y, y)$$

\noindent for every $x,y \in X_1$, then it is easy to show that
$$v_2(x) = -v_1(x) + c$$

\noindent for some constant $c$.  The function $v_1$ then gives a cardinal representation for a preference order on $X_1$ that ranks $x$ above $y$ if
$$(x, z) \succeq (y,z)$$

\noindent for any $z \in X_1$.  Our second use of Lemma \ref{lem:wakker} derives our full representation from the difference order on $S$.  Crucially, the difference order agrees with the additively separable utility we constructed on $\{s_0\} \times S_+$ for each action $s_0 \in S_0$.  Since these utilities are unique up to scale and location, the full representation is additively separable for each $s_0$.

\subsection{Proof of Theorem \ref{theo:main1}}\label{sec:proof1}

\hspace{1 pc}
Joint independence is clearly necessary, so we prove only sufficiency.  Since $\succeq_0$ satisfies opponent independence and the hexagon condition on $\{s_0\} \times S_+$ for each $s_0 \in S_0$, we can apply Lemma \ref{lem:wakker} on each restricted space $\{s_0\} \times S_+$ to obtain an additively separable utility representation holding $s_0$ fixed.  That is, for each $s_0 \in S_0$, we can find functions $\{g_i\}_{i = 1}^n$ such that
%
\begin{equation}\label{eq:addsepproof}
    \sum_{i = 1}^n g_i(s_0, s_i)
\end{equation}

\noindent represents $\succeq_0$ restricted to $\{s_0\} \times S_+$.  By Assumption \ref{as:separable1}, there exists $s_0^*$ with at least two essential opponents, and we know the representation is cardinal on $\{s_0^*\} \times S_+$.  We now construct a difference order $\succeq$ on $S$ in four steps:
    \begin{enumerate}
        \item Define $ss \simeq tt$ for every $s,t \in S$---this must hold in any difference order.

        \item Define $st \succeq s't$ and $ts' \succeq ts$ whenever $s \succeq_0 s'$---this ensures coordinate independence of the difference order as a relation on $S\times S$ and that the difference order represents the preferences $\succeq_0$.

        \item For each $s_0$ with at least two essential opponents, define $(s_0, s_+)(s_0, t_+) \succeq (s_0, s'_+)(s_0, t'_+)$ whenever
        $$\sum_{i=1}^n g_i(s_0, s_i) - g_i(s_0, t_i) \geq \sum_{i=1}^n g_i(s_0, s'_i) - g_i(s_0, t'_i).$$

        This ensures that $\succeq$ agrees with the additively separable representation on the restriction to $\{s_0\} \times S_+$ for each $s_0$.

        \item Take the transitive closure.
    \end{enumerate}

We show shortly that i) the order $\succeq$ is total on $S\times S$, ii) it represents the preferences $\succeq_0$ in the sense that $st \succeq s't$ if and only if $s\succeq_0 s'$, and iii) it is continuous and satisfies the hexagon condition.  Given these properties, we can apply Lemma \ref{lem:wakker} a second time to obtain continuous functions $v_1(s)$ and $v_2(s)$ such that
$$st \succeq s't' \quad \iff \quad v_1(s) + v_2(t) \geq v_1(s') + v_2(t').$$

\noindent Part (a) of the construction implies $v_2(s) = -v_1(s) + c$ for some constant $c$.  Since the representation of a difference order is unique up to positive affine transformations, we know from part (c) that
$$v_1(s_0, s_+) = c_1(s_0)\sum_{i = 1}^n g_i(s_0, s_i) + c_2(s_0)$$

\noindent for some functions $c_1$ and $c_2$ whenever there are at least two essential opponents at $s_0$---if there are $0$ or $1$ essential opponents, the utility trivially takes this form. Absorbing $c_2$ into one of the terms in the sum yields our desired representation.

It remains to establish properties i), ii), and iii) of the difference order.  Addressing property i) first, given any four profiles $s, t, s', t'$, we first use Assumption \ref{as:separable1} and find $s_0^*$ with at least two essential opponents.  By Assumption \ref{as:separable2} we can find $\hat{s}_+$, $\hat{t}_{+}$, $\hat{s}'_{+}$, and $\hat{t}'_+$ such that
$$s \simeq_0 (s_0^*, \hat{s}_+), \quad t \simeq_0 (s_0^*, \hat{t}_+), \quad s' \simeq_0 (s_0^*, \hat{s}'_{+}), \quad \text{and} \quad t' \simeq_0 (s_0^*, \hat{t}'_+).$$

\noindent Using part (b) of the construction and transitivity, we now have
$$st \simeq s (s_0^*, \hat{t}_+) \simeq (s_0^*, \hat{s}_+)(s_0^*, \hat{t}_+), \quad s't' \simeq s' (s_0^*, \hat{t}'_+) \simeq (s_0^*, \hat{s}'_+)(s_0^*, \hat{t}'_+).$$

\noindent Part (c) yields a comparison between $(s_0^*, \hat{s}_+)(s_0^*, \hat{t}_+)$ and $(s_0^*, \hat{s}'_+)(s_0^*, \hat{t}'_+)$, and transitivity implies the same comparison between $st$ and $s't'$.

For part ii), note that step (b) in the construction tells us that $st \succeq s't$ whenever $s \succeq_0 s'$---we just need to check that the comparisons we added in part (c) did not flatten strict comparisons of this form.  Equivalently, we need to show that the comparisons we derived in part i) do not depend on our choice of $s_0^*$.  Fix $s_0^*$ given by Assumption \ref{as:separable2}, and consider any other action $t_0 \in S_0$ with at least three essential opponents.  Fix opponent profiles $s_+, t_+, s'_+, t'_+$ such that
$$(s_0^*, s'_+) \simeq_0 (t_0, t'_+) \succ_0 (s_0^*, s_+) \simeq_0 (t_0, t_+).$$

\noindent Using positive affine transformations, we can rescale and translate the functions $\{g_i\}_{i=1}^n$ so that
$$\sum_{i=1}^n g_i(s_0^*, s_i) = \sum_{i=1}^n g_{i}(t_0, t_i) = 0, \quad \text{and} \quad \sum_{i=1}^n g_i(s_0^*, s'_i) = \sum_{i=1}^n g_{i}(t_0, t'_i) = 1.$$

\noindent We need to show that $(s_0^*, s''_+) \simeq_0 (t_0, t''_+)$ whenever
$$\sum_{i=1}^n g_i(s_0^*, s''_i) = \sum_{i=1}^n g_{i}(t_0, t''_i),$$

\noindent which implies that we get the same comparisons in part i) if we use $t_0$ instead of $s_0^*$.

Since the utility \eqref{eq:addsepproof} represents $\succeq_0$ when restricted to either $\{s_0^*\} \times S_+$ or $\{t_0\} \times S_+$, we know that our focal player is indifferent if either both sums are $0$ or both sums are $1$.  Moreover, for any other $x \in \bb{R}$, if there exist profiles $(s_0^*, \hat{s}_+) \simeq_0 (t_0, \hat{t}_+)$ such that
$$\sum_{i=1}^n g_i(s_0^*, \hat{s}_i) = \sum_{i=1}^n g_{i}(t_0, \hat{t}_i) = x,$$

\noindent then our focal player is always indifferent when both sums are equal to $x$.  Let $X$ denote the set of $x \in \bb{R}$ such that $(s_0^*, s''_+) \simeq_0 (t_0^*, t''_+)$ whenever
$$\sum_{i=1}^n g_i(s_0^*, s''_i) = \sum_{i=1}^n g_{i}(t_0, t''_i) = x.$$

\noindent We first show that $[0,1] \in X$.  We then show that $X$ contains the entire range of our utility on $\{s_0^*\} \times S_+$.
    
Focusing first on $[0,1]$, we show inductively that $X$ includes a grid with mesh size $\frac{1}{2^k}$ for each $k$, and continuity then implies that $X$ includes the entire interval.  Since there are at least three essential opponents at $s_0^*$, and each $g_i$ is continuous, we can modify the actions in the profile $s'_+$ so that
$$g_1(s_0^*, s'_1) + \sum_{i=2}^n g_i(s_0^*, s_i) = g_1(s_0^*, s_1) + \sum_{i = 2}^n g_i(s_0^*, s'_i) = \frac{1}{2}.$$

    
\noindent Without loss, we are assuming that opponent $1$ is essential at both $s_0^*$ and $t_0$, and note that if boundaries on the action sets $S_i$ would prevent such a construction, we could simply start from an $s'_+$ closer to $s_+$ for our initial scaling.  Similarly, we can modify the actions in $t'_+$ so that
$$g_1(t_0, t'_1) + \sum_{i=2}^n g_i(t_0, t_i) = g_1(t_0, t_1) + \sum_{i = 2}^n g_i(t_0, t'_i) = \frac{1}{2}.$$


\noindent Suppose that $(s_0^*, s_1, s'_{-1}) \succ_0 (t_0, t_1, t'_{-1})$, which implies also that $(s_0^*, s'_1, s_{-1}) \succ_0 (t_0, t'_1, t_{-1})$.  These two comparisons together with $(t_0, t_+) \succeq_0 (s_0^*, s_+)$ and $(t_0, t'_+) \succeq_0 (s_0^*, s'_+)$ violate joint order and strategic independence.  We can derive a similar violation if $(s_0^*, s_1, s'_{-1}) \prec_0 (t_0, t_1, t'_{-1})$, so we conclude that $(s_0^*, s_1, s'_{-1}) \simeq_0 (t_0, t_1, t'_{-1})$, and hence $\frac{1}{2} \in X$.  Analogous constructions then show that $\frac{1}{4}, \frac{3}{4} \in X$, and so on for the rest of our grid.


Having established that $X$ contains an interval, let $I$ denote a maximal interval in $X$, and note that continuity of $\succeq_0$ implies that $I$ is closed.  If $I$ does not cover the entire range of our utility, consider an endpoint of $I$---without loss take $b = \max \, I$, and suppose for any $\epsilon > 0$ there exists $x_\epsilon \in (b, b + \epsilon)$ with $x_\epsilon \notin X$.  Fix profiles $s^*_+$ and $t^*_+$ such that
$$\sum_{i=1}^n g_i(s_0^*, s^*_i) = \sum_{i=1}^n g_{i}(t_0, t^*_i) = x_\epsilon$$

\noindent For small enough $\epsilon$, we can find $\hat{s}_+$ and $\hat{t}_+$ such that
$$g_1(s_0^*, s^*_1) + \sum_{i = 2}^n g_i(s_0^*, \hat{s}_i) = g_1(s_0^*, \hat{s}_1) + \sum_{i = 2}^n g_i(s_0^*, s^*_i) = b, \text{ and}$$
$$g_1(t_0, t^*_i) + \sum_{i = 2}^n g_i(t_0, \hat{t}_i) = g_1(t_0, \hat{t}_i) + \sum_{i = 2}^n g_i(t_0, t^*_i) = b.$$


\noindent Since $b-\epsilon \in X$ for sufficiently small $\epsilon$, we have
$$(s_0^*, \hat{s}_+) \simeq_0 (t_0, \hat{t}_+), \quad (s_0^*, s^*_1, \hat{s}_{-1}) \simeq_0 (t_0, t^*_1, \hat{t}_{-1}), \quad \text{and} \quad (s_0^*, \hat{s}_1, s^*_{-1}) \simeq_0 (t_0, \hat{t}_{1}, t^*_{-1}),$$

\noindent and joint order and strategic independence implies $(s_0^*, s^*_+) \simeq_0 (t_0, t^*_+)$.  We conclude that $I$ was not maximal, and hence $I$ covers the entire range.

For property iii) we first focus on continuity.  We show that the upper contour sets of $\succeq$ in $S^2$, equipped with product topology, are open---analogously, lower contour sets are also open, implying that $\succeq$ is continuous.  That is, given an arbitrary pair $s't'$, we show that the set
$$U_{s' t'} := \{st \, : \, st \succ s' t'\}$$

\noindent is open in $S^2$.  Observe that for any $\hat{s}\hat{t}$, the set
$$\tilde{U}_{\hat{s}\hat{t}} := \{st \, : \, s \succ_0 \hat{s}, \, t \prec_0 \hat{t}\}$$

\noindent is open in $S^2$ since $\succeq_0$ is continuous.  Part (b) of our construction implies that $\tilde{U}_{\hat{s}\hat{t}} \subseteq U_{s' t'}$ for any $\hat{s}\hat{t} \succeq s' t'$.  We show that $U_{s' t'}$ is in fact the union of all such sets $\tilde{U}_{\hat{s} \hat{t}}$---since it is a union of open sets, it is open.

Fix $st$ such that $st \succ s' t'$, pick $s_0^*$ with three essential opponents, and find $\hat{s}_+, \hat{t}_+, \hat{s}'_+, \hat{t}'_+ \in S_+$ so that
$$s \simeq_0 (s_0^*, \hat{s}_+), \quad t \simeq_0 (s_0^*, \hat{t}_+), \quad s^* \simeq_0 (s_0, \hat{s}^*_+), \quad t^* \simeq_0 (s_0, \hat{t}^*_+).$$

\noindent We clearly have
$$(s_0^*, \hat{s}_+) (s_0^*, \hat{t}_+) \simeq st \succ s' t' \simeq (s_0^*, \hat{s}'_+)(s_0^*, \hat{t}'_+).$$

\noindent Since $\succeq$ is continuous on $\{s_0^*\} \times S_+$, we can find $(s_0^*,s''_+) \prec_0 s$ and $(s_0^*, t''_+) \succ_0 (s_0^*, \hat{t}_+) \simeq_0 t$ so that
$$(s_0^*,s''_+)(s_0^*, t''_+) \succeq (s_0^*, \hat{s}'_+)(s_0^*, \hat{t}'_+) \simeq s' t',$$

\noindent and we clearly have $st \in \tilde{U}_{(s_0^*, s''_+)(s_0^*, t''_+)}$ as required.
    
Finally, the hexagon condition follows because $\succeq$ must satisfy the hexagon condition when restricted to $\{s_0^*\} \times S_+$.  Suppose $st \simeq s't'$ and $s'' t' \simeq s't \simeq s t''$.  We can find $\hat{s}_+, \hat{t}_+, \hat{s}'_+, \hat{t}'_+, \hat{s}''_+, \hat{t}''_+$ such that
$$s \simeq_0 (s_0^*, \hat{s}_+), \, t \simeq_0 (s_0^*, \hat{t}_+), \, s' \simeq_0 (s_0^*, \hat{s'}_+), \,t' \simeq_0 (s_0^*, \hat{t}'_+), \,s'' \simeq_0 (s_0^*, \hat{s}''_+), \,t'' \simeq_0 (s_0^*, \hat{t}''_+).$$

\noindent The same equivalences then hold if we replace $s, t, s', t', s'', t''$ with the corresponding profiles involving $s_0^*$, and since the hexagon condition holds on $\{s_0^*\} \times S_+$ we have
$$(s_0^*, \hat{s}'_+)(s_0^*, \hat{t}''_+) \simeq (s_0^*, \hat{s}''_+)(s_0^*, \hat{t}_+) \quad \implies \quad s't'' \simeq s'' t$$

\noindent as required. \qedsymbol

\subsection{Discussion}

\hspace{1 pc}
Assumption \ref{as:separable2} is the most restrictive condition we use in our proof, but as already noted, we can substantially weaken it.  Suppose instead that we can cover $S$ using a collection of open subsets $\{U_k\}_{k = 1}^\infty$ that each individually satisfy Assumption \ref{as:separable2}.  We can then apply Theorem \ref{theo:main1} to each subset $U_k$.  Since the representation is unique up to positive affine transformations, when two subsets $U_k$ and $U_{k'}$ overlap, we can translate and rescale to ensure that the utilities match.  Inducting, we obtain an additively separable utility over all of $S$.  Nevertheless, we do need some assumption along these lines as the following example demonstrates.

\begin{example}\label{ex:obstruction}

    Suppose each player chooses a strictly positive, real-valued action, and our focal player strictly prefers $(s_0, s_+)$ to $(s'_0, s'_+)$ if either:
    \begin{enumerate}
        \item $s_0 \geq \frac{\pi}{2}$ and $s'_0 < \frac{\pi}{2}$, or

        \item $s_0, s'_0 \geq \frac{\pi}{2}$ and $\sum_{i=1}^n s_i > \sum_{i=1}^n s'_i$, or

        \item $s_0, s'_0 < \frac{\pi}{2}$ and $\tan(s_0) \sum_{i=1}^n s_i > \tan(s'_0) \sum_{i=1}^n s'_i$.
    \end{enumerate}

    \noindent One can readily verify that Assumption \ref{as:separable1} holds, as does joint independence.  Indeed from the specification above, we have additively separable representations on two disjoint domains, taking $u_0(\mathbf{s}) = \tan(s_0)\sum_{i=1}^n s_i$ on $(0, \frac{\pi}{2}) \times S_+$ and $u_0(\mathbf{s}) = \sum_{i=1}^n s_i$ on $[\frac{\pi}{2}, \infty) \times S_+$.  Nevertheless, no additively separable utility works on the entire set of strategy profiles $S$.
    
    To see why, consider the restriction to $\{1, \frac{\pi}{2}\}$, allowing our focal player just two actions.  Our focal player strictly prefers $(\frac{\pi}{2}, s_+)$ over $(1, s'_+)$ for any pair of opponent profiles, and holding $s_0$ fixed, she prefers a profile with a higher sum of opponent actions.  Restricted to either $\{1\} \times S_+$ or $\{\frac{\pi}{2}\} \times S_+$, Lemma \ref{lem:wakker} implies that any additively separable representation is a positive affine transformation of $\sum_{i=1}^n s_i$.  If we had an additively separable representation over the whole domain, its restrictions to these two regions would take the form
    $$\beta_1 + \alpha_1\sum_{i=1}^n s_i \hspace{0.5cm}\text{ and }\hspace{0.5cm}
        \beta_{\pi/2} + \alpha_{\pi/2}\sum_{i=1}^n s_i$$

    \noindent for constants $\beta_1, \beta_{\pi/2}$ and $\alpha_1, \alpha_{\pi/2} > 0$.  However, since sums are unbounded, given any such constants, we can find $s_+$ and $s'_+$ such that
    $$\beta_1 + \alpha_1\sum_{i=1}^n s_i > \beta_{\pi/2} + \alpha_{\pi/2}\sum_{i=1}^n s'_i,$$

    \noindent which violates the preference for $s_0 = \frac{\pi}{2}$ over $s_0 = 1$.

\end{example}

\section{Strategically Separable Preferences}\label{sec:stratseparable}

\hspace{1 pc}
Many preferences that appear in network games are not separable.  For instance, a popular way to generalize the classic linear-quadratic utility assumes that preferences depend on a linear aggregate of neighbors' actions---actions are real-valued, and player $i$'s payoff is
\begin{equation}\label{eq:linagg}
u_i(\mathbf{s}) = b(s_i) + \gamma(s_i) \psi\left(\sum_{j \neq i} g_{ij} s_j\right)
\end{equation}

\noindent for some functions $b$, $\gamma$, and $\psi$.  Taking $b(s_i) = b_i s_i-\frac{s_i^2}{2}$ and $\gamma(s_i) =  s_i$, this form includes preferences in recent work by \citet{Belhajetal2014,PariseOzdaglar2019},and \citet{ZenouZhou2023}.  Such preferences are typically not equivalent to a utility of the form \eqref{eq:addseputility}, but they often \emph{are} strategically separable.

Recall our example from the introduction in which $S_i = \bb{R}_+$ for each player $i$, and player $i$'s payoffs are
$$u_i(s) = b_i s_i - s_i^2 + s_i \ln \left(\sum_{j \in G_i} g_{ij} s_j\right).$$

\noindent Although this utility satisfies opponent independence, joint independence fails in general.  To see this, suppose our focal player $0$ faces just two opponents and set $b_0 = 0.15$.  In Definition \ref{def:jointind}, take $s_0 = 1.2$, $s'_0 = 1$, $s_1 = s'_1 = 0.5$, $s_2 = 1$, $s'_2 = 0.5$, $t_1 = t'_1 = 2$, $t_2 = 2.5$, and $t'_2 = 2$.  With these values, we have
$$(s_0, s_1, s_2) \succeq_0 (s'_0, s'_1, s'_2) \quad \text{and} \quad(s'_0, t'_1, s'_2) \succeq (s_0, t_1, s_2)$$

\noindent but
$$(s'_0, s'_1, t'_2) \succeq (s_0, s_1, t_2) \quad \text{and} \quad (s_0, t_1, t_2) \succeq_0 (s'_0, t'_1, t'_2).$$

\noindent Nevertheless, the corresponding family of orders $\{\succeq_{s_{-i}}\}_{s_{-i} \in S_{-i}}$ satisfies strategic independence, and we already saw that the above utility is strategically equivalent to
$$v_i(s) = - e^{s_i} - e^{b_i-s_i} \sum_{j \in G_i} g_{ij} s_j.$$





\noindent Our main result in this section shows that, under somewhat stronger technical restrictions than what we imposed earlier, strategic independence of the family $\{\succeq_{s_+}\}_{s_+ \in S_+}$ ensures that our focal player's preferences are strategically separable.  Utilities of the form \eqref{eq:linagg} that depend on linear aggregates generally satisfy our conditions as long as $\psi$ is monotonic.  Indeed, such preferences imply additional structure, a constant rate of substitution between neighbors, giving the representation of Theorem \ref{theo:CRS}.



The proof of Theorem \ref{theo:main2} follows the same two-step outline as that of Theorem \ref{theo:main1}.  We first construct additively separable utilities on restricted spaces before showing that we can merge them to get a utility over the whole space.  Both steps present new challenges.  In the first, since our primitive does not include preferences over opponent actions, we cannot immediately apply Lemma \ref{lem:wakker}.  Instead, for each pair of actions $s_0, s'_0 \in S_0$, we construct a coordinate independent order $\succeq$ on $S_+$ with the property that, if $s'_0 \succeq_{s_+} s_0$, then $s'_0 \succeq_{s'_+} s_0$ for all $s'_+ \succeq s_+$.  We then apply Lemma \ref{lem:wakker} to get an additively separable utility for each \emph{pair} of actions $s_0, s'_0 \in S_0$.  In the second step, we construct a difference order on $S$ to merge our pairwise utilities, but without direct preference comparisons between opponent profiles, we need additional structure.  The key tool for this step is a new property, \emph{joint solvability}, which is similar in spirit to requiring that payoffs depend on an aggregate of opponent actions.

In the next subsection, we detail technical assumptions.  We then state our result and prove a lemma that constitutes the first step of our outline---we relegate the second step to an appendix.  We subsequently discuss prospects for weakening technical assumptions.

\subsection{Technical Assumptions}\label{sec:assumptions}

\hspace{1 pc}
We divide our assumptions into two rough categories:  those that involve topological properties of action sets and preferences, and those that require sufficiently ``rich'' action sets.  We begin with topological properties.

\begin{definition}\label{def:jointcont}
The family of orders $\{\succeq_{s_+}\}_{s_+ \in S_+}$ is \textbf{jointly continuous} if for any $s^*_0 \in S_0$, whenever we have $s_0 \succ_{s_+}  s^*_0$ ($s_0 \prec_{s_+} s^*_0$), there exists an open neighborhood $U_0$ of $s_0$, and an open neighborhood $U_+$ of $s_+$, such that $t_0 \succ_{t_+} s^*_0$ ($t_0 \prec_{t_+} s^*_0$)for all $t_0 \in U_0$ and $t_+ \in U_+$.

\vspace{1 pc}

\noindent The family has \textbf{connected intervals} if for any two actions $s_0, s'_0 \in S_0$, any opponent $i$, and any two profiles $s_{-i}, s'_{-i}$ for opponents other than $i$, the set
$$\{s_i \, : \, s_0 \succeq_{s_i, s_{-i}} s'_0, \quad s'_0 \succeq_{s_i, s'_{-i}} s_0\}$$

\noindent is connected in $S_i$.
\end{definition}

Joint continuity of the family $\{\succeq_{s_+}\}_{s_+ \in S_+}$ is analogous to assuming continuity of a preference order $\succeq_0$ on $S$.  The definition implies that each $\succeq_{s_+}$ is a continuous order on $S_0$---if $s_0 \succ_{s_+} (\prec_{s_+}) \, s^*_0$, then $t_0 \succ_{s_+} (\prec_{s_+}) \, s^*_0$ for all $t_0 \in U_0$.  Requiring the strict comparison to continue to hold in an open neighborhood of both $s_0$ and $s_+$ is precisely what we need to ensure that the utility we obtain is continuous on $S = S_0 \times S_+$.

Connected intervals is a new condition, but it holds for all standard utilities in network games.  For instance, if our focal player has preferences that follow a linear-quadratic utility
$$u(s) = b s_0 - \frac{s_0^2}{2} + s_0 \sum_{j=1}^n \gamma_j s_j$$

\noindent then for $s_0 > s'_0$, our focal player prefers $s_0$ exactly when
$$\sum_{j=1}^n \gamma_j s_j \geq \frac{s_0 + s'_0}{2} - b_i.$$

\noindent Fixing an opponent $s_i$ and profiles $s_{-i}, s'_{-i}$ for the other opponents, these preferences have connected intervals if the set of $s_i$ such that
$$\gamma_i s_i + \sum_{j \neq i}\gamma_j s_j \geq \frac{s_0 + s'_0}{2} - b_i \quad \text{and} \quad \gamma_i s_i + \sum_{j \neq i}\gamma_j s'_j \leq \frac{s_0 + s'_0}{2} - b_i$$

\noindent is connected.  This set is exactly the interval between
$$\frac{1}{\gamma_i}\left(\frac{s_0 + s'_0}{2} - b_i - \sum_{j \neq i}\gamma_j s_j\right) \quad \text{and} \quad \frac{1}{\gamma_i}\left(\frac{s_0 + s'_0}{2} - b_i - \sum_{j \neq i}\gamma_j s'_j\right),$$

\noindent which is connected.\footnote{If $\sum_{j \neq i} \gamma_j s'_j > \sum_{j \neq i} \gamma_j s_j$, the set is empty and therefore trivially connected. Otherwise, it is a non-empty closed interval.}  More intuitively, what this condition means is that, in an additively separable representation \eqref{eq:addseputility}, the set of $s_i$ such that
$$g_i(s_0, s_i) - g_i(s'_0, s_i) \in [a,b]$$

\noindent is a connected set in $S_i$ for any interval $[a,b]$.  In one-dimensional action spaces, this is equivalent to requiring that the difference above is weakly monotone in $s_i$.


To ensure that we have enough preference comparisons to pin down our representation, we also need to impose richness conditions on the action sets.

\begin{definition}\label{def:richness2}
    Given a family of orders $\{\succeq_{s_+}\}_{s_+ \in S_+}$, opponent $i$ is \textbf{essential} if for any two actions $s_0, s'_0 \in S_0$, there exists a profile of actions $s_{-i}$ for other opponents and two actions $s_i, s'_i \in S_i$ such that
    $$s_0 \succeq_{s_i, s_{-i}} s'_0 \quad \text{and} \quad s'_0 \succeq_{s'_i, s_{-i}} s_0,$$

    \noindent with at least one strict inequality.  Opponent $i$ is \textbf{completely essential} if additionally, for any two actions $s_0, s'_0 \in S_0$ and any $s_{-i} \in S_i$, we can find $s_i \in S_i$ such that
    $$s_0 \simeq_{s_i, s_{-i}} s'_0.$$
    
    \noindent The family of orders $\{\succeq_{s_+}\}_{s_+ \in S_+}$ satisfies \textbf{joint solvability} if for any two completely essential opponents $i,j$, any three actions $s_0, s'_0, s''_0$, any $s_i, s'_i \in S_i$, and $s_j \in S_j$, whenever we have
    $$s_0 \simeq_{s'_i, s_j, s_{-ij}} s'_0, \quad s_0 \simeq_{s'_i, s_j, s'_{-ij}} s''_0, \quad \text{and} \quad s'_0 \simeq_{s'_i, s_j, s''_{-ij}} s''_0$$
    
    \noindent for some profiles $s_{-ij}, s'_{-ij}, s''_{-ij}$ of actions for opponents other than $i$ and $j$, then we can find an action $s'_j \in S_j$ such that
    $$s_0 \simeq_{s_i, s'_j, s_{-ij}} s'_0, \quad s_0 \simeq_{s_i, s'_j, s'_{-ij}} s''_0, \quad \text{and} \quad s'_0 \simeq_{s_i, s'_j, s''_{-ij}} s''_0.$$
\end{definition}

This definition of an essential opponent mirrors the earlier one.  An opponent $i$ is essential if at some profile $s_{-i}$ for opponents other than $i$, opponent $i$'s action can flip our focal player's preference between $s_0$ and $s'_0$.  Here, we require this opponent's action to affect all pairs of comparisons, rather than just one---we could weaken this, but it would complicate the other definitions.  Opponent $i$ is completely essential if, no matter what the other opponents are doing, her choice can still change our focal player's preference between $s_0$ and $s'_0$.  Effectively, this means that the differences
$$g_i(s_0, s_i) - g_i(s'_0, s_i)$$

\noindent are unbounded on both sides in the representation we derive.  While many network games do not have any completely essential opponents---for instance, in linear-quadratic models we typically assume actions are non-negative---we can often extend action sets and preferences in a natural way (e.g., let $S_i = \bb{R}$ in a linear-quadratic game, and keep the same utility function).  To allow a straightforward application of Lemma \ref{lem:wakker}, we require that our focal player has at least three completely essential opponents.

Joint solvability is a more substantive condition.  This says that, however the change from $s_i$ to $s'_i$ affects our focal player's preferences between three actions, we can achieve the same effect on all three comparisons by changing a different opponent's action $s_j$ to some $s'_j$.  This weakens a key property that utilities based on linear aggregates enjoy.  If actions are real-valued, and we have
$$u(s) = f\left(s_0, \sum_{i=1}^n \gamma_i s_i\right),$$

\noindent then we can replicate the impact of changing $s_i$ to $s'_i$ by changing $s_j$ to
$$s'_j = s_j + \frac{\gamma_i}{\gamma_j}(s'_i - s_i).$$

\noindent In this case, our focal player's entire preference order is preserved if we switch from $(s'_i, s_j)$ to $(s_i, s'_j)$.  Joint solvability holds more generally, including for preferences that depend on arbitrary non-linear aggregates.  In our proof, we first obtain an additively separable utility \emph{for each pair of actions} $s_0$ and $s'_0$, and this step does not depend at all on joint solvability.  Joint solvability, together with connected intervals, ensures that we can merge these pairwise utilities into a single additively separable utility over all of $S$.

Assumption \ref{as:basic2} recapitulates all of the technical conditions we need.

\begin{assumption}\label{as:basic2}
    Each action set $S_i$ is a connected topological spaces.  The family of orders $\{\succeq_{s_+}\}_{s_+ \in S_+}$ is jointly continuous and has connected intervals.  At least three opponents are completely essential, and the family satisfies joint solvability.
\end{assumption}

\subsection{The Result}

\hspace{1 pc}
Given our technical assumptions on action sets and preferences, strategic independence is precisely the condition needed to ensure strategic separability.

\begin{Theorem}\label{theo:main2}
    Suppose Assumption \ref{as:basic2} holds for action sets $\{S_i\}_{i=0}^n$ and the family of preference orders $\{\succeq_{s_+}\}_{s_+ \in S_+}$.  The family is strategically separable if and only if it satisfies strategic independence.  Moreover, the representation of $\{\succeq_{s_+}\}_{s_+ \in S_+}$ via a utility of the form
    $$u(s) = \sum_{i=1}^n g_i(s_0, s_i)$$

    \noindent is unique up to a positive affine transformation.
\end{Theorem}

As before, a utility of the form \eqref{eq:addseputility} clearly implies strategic independence, so we prove only the ``if'' direction.  Mirroring our proof of Theorem \ref{theo:main1}, we construct a difference order on $S$ that is consistent with the preferences $\{\succeq_{s_+}\}_{s_+ \in S_+}$ on $S_0$ and with a coordinate independent order on $S_+$ for each fixed $s_0$.  Applying Lemma \ref{lem:wakker} then gives our required representation.  The following lemma provides the first step in our argument.  Unlike in the last section, we cannot immediately obtain an additively separable utility for each fixed $s_0$ because we do not have preferences over opponent actions.  Instead, for each pair of actions $s_0$ and $s'_0$, we construct a continuous additively separable utility that represents our focal player's preferences between those two actions.

\begin{lemma}\label{lem:pairwise}
    Suppose each action set $S_i$ is a connected topological space, the family of orders $\{\succeq_{s_+}\}_{s_+ \in S_+}$ is jointly continuous, and at least three opponents are completely essential.  If the family satisfies strategic independence, then for any pair of actions $s_0, s'_0 \in S_0$, there exist continuous functions $\{g_i\}_{i = 1}^n$ such that $s'_0 \succeq_{s_+} s_0$ if and only if
    $$\sum_{i=1}^n g_i(s_i) \geq 0.$$

    \noindent Moreover, these functions are unique up to a positive rescaling.
\end{lemma}


\begin{proof}
We fix two actions $s_0$ and $s'_0$ and define an order $\succeq$ on $S_+$.  Our interpretation of $\succeq$ is that $s'_{+} \succeq s_{+}$ if $s'_0$ is relatively more attractive to our focal player when opponents play $s'_{+}$ instead of $s_{+}$.

Formally, for each opponent $i$ and any two profiles $s_{-i}, s'_{-i} \in S_{-i}$, we define $(s_i, s'_{-i}) \succeq (s_i, s_{-i})$ for \emph{every} $s_i \in S_i$ if there exists $s^*_i$ such that
$$s_0 \succeq_{(s^*_i, s_{-i})} s'_0, \quad \text{and} \quad s'_0 \succeq_{(s^*_i, s'_{-i})} s_0.$$

\noindent That is, if switching from $s_{-i}$ to $s'_{-i}$ can ever flip our focal player's preference in favor of $s'_0$, we define $(s_i, s'_{-i}) \succeq (s_i, s_{-i})$ for all $s_i$---note the order is coordinate independent by construction.

We need to establish three properties:  i) after taking the transitive closure we have a total order on $S_+$, ii) the order is continuous, and iii) if $s'_0 \succeq_{s_+} s_0$, then $s'_0 \succeq_{s'_+} (\succ_{s'_+}) s_0$ for every $s'_+ \succeq (\succ) s_+$.  Given these properties, we can apply Lemma \ref{lem:wakker} to obtain continuous functions $\{g_i\}_{i=1}^n$ such that
$$s'_+ \succeq s_+ \quad \iff \quad \sum_{i=1}^n g_i(s'_i) \geq \sum_{i=1}^n g_i(s_i).$$

\noindent Moreover, these functions are unique up to a positive affine transformation.  Property iii) implies that we can translate the functions $\{g_i\}$ such that $\sum_{i=1}^n g_i(s_i) = 0$ whenever $s'_0 \simeq_{s_+} s_0$, and we have $s'_0 \succ_{s_+} (\prec_{s_+}) s_0$ whenever the sum is positive (negative).  This gives the desired representation of our focal player's preference between $s'_0$ and $s_0$.

We now establish properties i), ii), and iii).  For property i), note that since each opponent $i$ is completely essential, the comparisons we explicitly added include comparisons between any $s_+$ and $s'_+$ that share at least one coordinate.  These comparisons are also already transitive on each $\{s_i\} \times S_{-i}$.  If $(s_i,s''_{-i}) \succeq (s_i, s'_{-i}) \succeq (s_i, s_{-i})$, then there exist $s^*_i$ and $s^{**}_i$ such that
$$s_0 \succeq_{s^*_i, s_{-i}} s'_0, \quad s'_0 \succeq_{s^*_i, s'_{-i}} s_0, \quad s_0 \succeq_{s^{**}_i, s'_{-i}} s'_0, \quad \text{and} \quad s'_0 \succeq_{s^{**}_i, s''_{-i}} s_0.$$

\noindent Given the first three comparisons above, strategic independence implies that $s_0 \succeq_{s^{**}_i, s_{-i}} s'_0$, and this together with the last comparison reveals that $s''_{-i} \succeq s_{-i}$.

To get a comparison between arbitrary $s_+$ and $s'_+$, we can find $(s_1, s_2^*, s'_{-12}) \simeq s'_+$, and transitivity implies that whatever comparison exists between $s_{+}$ and $(s_1, s_2^*, s'_{-12})$ must also exist between $s_+$ and $s'_+$.  Property ii) follows from joint continuity of the family $\{\succeq_{s_+}\}_{s_+ \in S_+}$.   

For property iii), suppose we have $s'_0 \succeq_{s_+} s_0$, but $s'_0 \prec_{s'_+} s_0$ for some $s'_+ \succeq s_+$---the other case is similar.  If $s'_+$ and $s_+$ share a coordinate, say $s_i = s'_i$, then since $s'_+ \succeq s_+$, we can find $s_i^*$ such that
$$s_0 \succeq_{s_i^*, s_{-i}} s'_0 \quad \text{and} \quad s'_0 \succeq_{s_i^*, s'_{-i}} s_0,$$

\noindent but together with
$$s'_0 \succeq_{s_i, s_{-i}} s_0 \quad \text{and} \quad s_0 \succ_{s_i, s'_{-i}} s'_0$$

\noindent this violates strategic independence.  If $s'_+$ and $s_+$ do not share a coordinate, we can reduce our problem to the previous case.  Since opponents are completely essential, we can find $s_{-1}^*$ so that
$$s'_0 \simeq_{s'_1, s'_2, s^*_{-12}} s_0 \quad \text{and} \quad s'_0 \simeq_{s_1, s_{-1}^*} s_0,$$

\noindent which implies $s'_+ \simeq (s_1, s^*_2, s'_{-12})$, so $(s_1, s^*_2, s'_{-12}) \succeq s_+$.  If $s'_0 \prec_{s_1, s^*_2, s'_{-12}} s_0$, then this conflicts with strategic independence as shown before.  If $s'_0 \succeq_{s_1, s^*_2, s'_{-12}} s_0$, then this comparison together with 
$$s'_0 \succeq_{s'_1, s'_2, s^*_{-12}} s_0, \quad s_0 \succeq_{s_1, s_{-1}^*} s_0, \quad \text{and} \quad s_0 \succ_{s'_+} s'_0$$

\noindent violates strategic independence.  We conclude the property iii) holds, completing the proof.

\end{proof}


We present the remainder of our proof in Appendix \ref{ap:proofmain}.  While Theorem \ref{theo:main2} requires stronger technical conditions than our characterization of separable preferences, it still covers important cases.  In particular, Theorem \ref{theo:main2} applies to preferences that depend on aggregates of opponent actions, which we explore in our next subsection.  Later, in Section \ref{sec:cyclic}, we give an alternative characterization that sheds further light on the tradeoffs required to ensure that strategic independence is sufficient to get an additively separable utility.






\subsection{Local Aggregates}

\hspace{1 pc}
Recall our example from the introduction:
$$u_i(s) = b_is_i - s_i^2 + s_i \ln \left(\sum_{j \in G_i} g_{ij} s_j\right).$$

\noindent A key feature is that player $i$'s utility depends on an aggregate of her neighbors' actions, which implies joint solvability.  If neighbor $j$ reduces her action by $1$, we can increase neighbor $k$'s action by $\frac{g_{ij}}{g_{ik}}$ to restore $i$'s preferences over $S_i$ to what they were.  The example thus illustrates a broader point:  generically, preferences that depend on aggregates of neighbors' actions are strategically separable.

\begin{Cor}\label{cor:addagg}
Suppose $(N, \{S_i\}_{i \in N}, \{u_i\}_{i \in N})$ is a normal form game with real-valued actions, and player $i$ has a continuous utility function
$$u_i(s) = u\left(s_i, \sum_{j \neq i} f_j(s_j)\right),$$

\noindent in which each $f_j$ is a monotone function of $s_j$.  Suppose further that player $i$ has no weakly dominated actions, and for any pair $s_i, s'_i \in S_i$, there is a threshold $\tau$ such that $u(s_i, x) \geq u(s'_i, x)$ precisely when either $x \geq \tau$ or $x \leq \tau$.  If $f_j(s_j)$ ranges over all of $\bb{R}$ for at least three opponents $j$, then player $i$'s preferences are strategically separable.
\end{Cor}

\begin{proof}
    The family of orders $\{\succeq_{s_{-i}}\}_{s_{-i} \in S_{-i}}$ is jointly continuous because $u_i$ is continuous, and it has connected intervals because $i$'s preference between any $s_i$ and $s'_i$ depends on whether the aggregate is above or below a threshold.  Since $f_j(s_j)$ ranges over all of $\bb{R}$ for at least three opponents $j$, we have three opponents who are completely essential, and joint solvability holds.  The result follows from Theorem \ref{theo:main2}.
\end{proof}

Corollary \ref{cor:addagg} highlights that many preferences in models of network games, among others, are strategically separable and can thus be represented via additively separable utilities of the form \eqref{eq:addseputility}.  When the aggregate is linear, as in our example, Theorem \ref{theo:CRS} applies, and our representation is necessarily linear in opponent actions.

\section{Proofs of Theorems \ref{theo:CRS} and \ref{theo:linquad}}\label{sec:agg}

\hspace{1 pc}
Having now addressed conditions that lead to separability and strategic separability, we return to the Theorems \ref{theo:CRS} and \ref{theo:linquad}.  Recall that for both of these results, we assume actions are real-valued, and preferences are strategically separable.  Theorem \ref{theo:CRS} asserts that if players preferences feature a constant rate of substitution between neighbors, then the additively separable utility is linear in opponent actions:
$$u_i(\mathbf{s}) = b_(s_i) + \gamma(s_i) \sum_{j \neq i} g_{ij} s_j.$$

\noindent Theorem \ref{theo:linquad} shows that, if additionally best responses are linear, and preferences satisfy midpoint indifference, we get the linear-quadratic form
$$u_i(\mathbf{s}) = b s_i + s_i\sum_{j \neq i} g_{ij} s_j - \frac{s_i^2}{2}.$$


\begin{proof}[Proof of Theorem \ref{theo:CRS}]
    By assumption, player $i$'s preferences are strategically equivalent to
    $$u_i(\mathbf{s}) = \sum_{j \neq i} g_{ij}(s_i, s_j)$$

    \noindent for some functions $\{g_{ij}\}$.  Given a constant rate of substitution between $j$ and $k$, for any $x \in \bb{R}$ we can also represent $i$'s preferences via
    $$u_i^x(\mathbf{s}) = g_{ij}(s_i, s_j - \delta x) + g_{ik}(s_i, s_k + x) + \sum_{\ell \neq i,j,k} g_{i\ell}(s_i, s_\ell)$$

    \noindent for some constant $\delta$.  Since the representation of the difference $u_i(s'_i, s_{-i}) - u_i(s_i, s_{-i})$ is unique up to scaling, we have
    \begin{align*}
        g_{ij}(s'_i, s_j) &- g_{ij}(s_i, s_j) + g_{ik}(s'_i, s_k) - g_{ik}(s_i, s_k) \\
        &=  g_{ij}(s'_i, s_j- \delta x) - g_{ij}(s_i, s_j-\delta x) + g_{ik}(s'_i, s_k + x) - g_{ik}(s_i, s_k + x)
    \end{align*}
    
    \noindent for all $s_i, s'_i, s_j, s_k, x$.  Fix $s'_i = s_i^0$ and $s_j = s_j^0$, and rearrange to get
    \begin{align*}
        g_{ik}(s_i, s_k+x) &- g_{ik}(s_i^0, s_k + x) + g_{ik}(s_i^0, s_k) - g_{ik}(s_i, s_k) \\
        &= g_{ij}(s_i, s_j^0) - g_{ij}(s_i^0, s_j^0) + g_{ij}(s_i^0, s_j^0 - \delta x) - g_{ij}(s_i, s_j^0- \delta x).
    \end{align*}

    \noindent This holds for all $s_i$, $s_k$, and $x$.  In particular, the left hand side is constant in $s_k$, so we can rewrite this as
    $$g_{ik}(s_i, s_k + x) - g_{ik}(s_i^0, s_k+x) = g_{ik}(s_i, s_k) - g_{ik}(s_i^0, s_k) + h(s_i, x)$$

    \noindent for some function $h(s_i, x)$.
    
    Now define
    $$b_k(s_i) = g_{ik}(s_i, 0) - g_{ik}(s_i^0, 0) \quad \text{and} \quad \tilde{g}_{ik}(s_i, s_k) = g_{ik}(s_i, s_k) - g_{ik}(s_i^0, s_k) - b_k(s_i),$$
    
    \noindent and substituting these definitions above gives
    $$\tilde{g}_{ik}(s_i, s_k +x) = \tilde{g}_{ik}(s_i, s_k) + h(s_i, x).$$

    \noindent By definition we have $\tilde{g}_{ik}(s_i, 0) = 0$ for all $s_i$, and substituting this gives $h(s_i, x) = \tilde{g}_{ik}(s_i, x)$.  We therefore arrive at Cauchy's functional equation in the second variable:
    $$\tilde{g}_{ik}(s_i, s_k +x) = \tilde{g}_{ik}(s_i, s_k) + \tilde{g}_{ik}(s_i, x)$$

    \noindent for all $s_i, s_k, x$.  Since all of the functions are continuous, the only solution is
    $$\tilde{g}_{ik}(s_i, s_k) = \gamma_{ik}(s_i) s_k $$

    \noindent for some function $\gamma_{ik}(s_i)$.  Taking $f_k(s_k) = g_{ik}(s_i^0, s_k)$, it follows that
    $$g_{ik}(s_i, s_k) = b_k(s_i) + f_k(s_k) + \gamma_{ik}(s_i)s_k .$$

    \noindent Repeating the argument for $j$, replacing $x$ with $-\frac{x}{\delta}$, gives
    $$g_{ij}(s_i, s_j) = b_j(s_i) + f_j(s_j) + \gamma_{ij}(s_i) s_j $$

    \noindent for some $b_j, f_j, \gamma_{ij}$.

    It remains to show that $\gamma_{ij}(s_i)$ is a constant multiple of $\gamma_{ik}(s_i)$.  Recall that we had
    \begin{align*}
        g_{ij}(s'_i, s_j) &- g_{ij}(s_i, s_j) + g_{ik}(s'_i, s_k) - g_{ik}(s_i, s_k) \\
        &=  g_{ij}(s'_i, s_j- \delta x) - g_{ij}(s_i, s_j-\delta x) + g_{ik}(s'_i, s_k + x) - g_{ik}(s_i, s_k + x)
    \end{align*}

    \noindent for any value of $x$.  Substituting our expressions for $g_{ij}$ and $g_{ik}$ and canceling common terms, we get
    $$\left(\gamma_{ik}(s'_i) - \gamma_{ik}(s_i)\right)x = \left(\gamma_{ij}(s'_i) - \gamma_{ij}(s_i)\right)\delta x,$$

    \noindent from which the conclusion follows.
\end{proof}

Together Theorems \ref{theo:CRS} and \ref{theo:main2} show that two key properties---strategic independence and a constant rate of substitution---combine to give an especially convenient canonical form for players' utility.  Almost all utilities used in models of network games satisfy these properties and are therefore strategically equivalent to utilities
$$u_i(\mathbf{s}) = b(s_i) + \gamma(s_i) \sum_{j \neq i} g_{ij} s_j.$$

\noindent This class offers a principled way to build on the standard linear-quadratic framework.  An important caveat is that strategic equivalence does not justify using utilities of this form for welfare comparisons.  We can obtain robust results on equilibria and comparative statics, but other utility specifications may make more sense for welfare analysis.

\begin{proof}[Proof of Theorem \ref{theo:linquad}]
    Given strategic separability and continuous CRS preferences, by Theorem \ref{theo:CRS} we can assume a representation
    $$u_i(\mathbf{s}) = b(s_i) + \gamma(s_i) \sum_{j \neq i} g_{ij} s_j.$$

    \noindent Assuming linear best responses, we can rescale the constants $g_{ij}$ so that best replies take the form $s_i^*(s_{-i}) = a + \sum_{j \neq i} g_{ij} s_j$.  The action $s_i$ is then a best response exactly when opponents choose actions such that
    $$\sum_{j \neq i} g_{ij} s_j = s_i - a.$$

    \noindent Given any three action $s_i$, $s'_i$, and $s''_i$, midpoint indifference gives us the following three equalities:
    $$b(s_i) + \gamma(s_i) \left(\frac{s_i + s'_i}{2} - a\right) = b(s'_i) + \gamma(s'_i) \left(\frac{s_i + s'_i}{2} - a\right)$$
    $$b(s'_i) + \gamma(s'_i) \left(\frac{s'_i + s''_i}{2} - a\right) = b(s''_i) + \gamma(s''_i) \left(\frac{s'_i + s''_i}{2} - a\right)$$
    $$b(s_i) + \gamma(s_i) \left(\frac{s_i + s''_i}{2} - a\right) = b(s''_i) + \gamma(s''_i) \left(\frac{s_i + s''_i}{2} - a\right).$$

    \noindent Adding the first two equalities together and rearranging gives
    $$b(s_i)-b(s''_i) = \gamma(s''_i) \left(\frac{s'_i + s''_i}{2} - a\right) + \gamma(s'_i) \frac{s_i - s''_i}{2} - \gamma(s_i)\left(\frac{s_i + s'_i}{2} - a\right),$$

    \noindent and using this together with the third equality, we have
    \begin{align*}
        b(s_i) - b(s''_i) &= \left(\gamma(s''_i) - \gamma(s_i)\right)\left(\frac{s_i + s''_i}{2} - a\right) \\
        &= \gamma(s''_i) \left(\frac{s'_i + s''_i}{2} - a\right) + \gamma(s'_i) \frac{s_i - s''_i}{2} - \gamma(s_i)\left(\frac{s_i + s'_i}{2} - a\right).
    \end{align*}

    \noindent After canceling terms and simplifying, we have
    $$\gamma(s''_i) (s_i - s'_i) + \gamma(s_i) (s'_i - s''_i) + \gamma(s'_i) (s''_i - s_i) = 0.$$

    \noindent Assuming distinct values of $s_i$, $s'_i$, and $s''_i$, we can divide by all three differences and rearrange to get
    $$\frac{\gamma(s''_i) - \gamma(s'_i)}{(s'_i - s''_i)(s''_i - s_i)} - \frac{\gamma(s'_i) - \gamma(s_i)}{(s_i - s'_i)(s''_i - s_i)} = \frac{\gamma(s''_i) - \gamma(s'_i)}{s'_i - s''_i} - \frac{\gamma(s'_i) - \gamma(s_i)}{s_i - s'_i} = 0$$

    \noindent for every $s_i, s'_i, s''_i$.  In particular, this means that $\frac{\gamma(s'_i) - \gamma(s_i)}{s'_i - s_i}$ is a constant, implying $\gamma$ is a linear function. Without loss, take $\gamma(s_i) = ms_i$,\footnote{Adding a constant to $\gamma(s_i)$ does not impact strategic equivalence.} and we substitute back into our earlier expression to get
    \begin{align*}
    b(s_i) - b(s''_i) &= \left(\gamma(s''_i) - \gamma(s_i)\right) \left(\frac{s_i + s''_i}{2} - a\right) \\
    &= m (s''_i - s_i) \left(\frac{s_i + s''_i}{2} - a\right) \\
    &= m\left(a s_i -\frac{1}{2}s_i^2 - a s''_i + \frac{1}{2} (s''_i)^2\right).
    \end{align*}

\noindent From this we conclude that $b(s_i) = m\left(a s_i - \frac{s_i^2}{2}\right)$, giving the linear-quadratic form---note we can rescale the utility to eliminate the constant $m$.
\end{proof}

Building on Theorem \ref{theo:CRS}, Theorem \ref{theo:linquad} delivers an exact characterization of the standard linear-quadratic utility.  In addition to separability, a linear-quadratic utility assumes a constant rate of substitution between neighbors, a linear best-response map, and midpoint indifference.

\section{Beyond Strategic Independence}\label{sec:cyclic}

\hspace{1pc}
Getting to an additively separable utility using strategic independence alone requires several assumptions on action sets and preferences that may not always hold.  In the absence of such restrictions, when are preferences strategically separable?  In this section, we provide a tight characterization using the concept of a \emph{balanced sequence}. Our approach applies equally well, after straightforward modifications, to characterize separability.  To avoid redundancy, we present results only for strategic separability.


\begin{definition}\label{def:cyclicsequence}
Given a family of preference orders $\{\succeq_{s_+}\}_{s_+ \in S_+}$, a \textbf{balanced sequence} is a collection of strategy profile pairs $\{(r^k, t^k)\}_{k=1}^m \subseteq S \times S$ such that
\begin{enumerate}
\item $r^k_+ = t^k_+$ for each $k$

\item $r^k_0 \succeq_{r^k_{+}} t^k_0$ for each $k$, with strict preference for at least one $k$

\item For each opponent $i$, each pair of actions $(s_0, s_i)$ appears the same number of times in the sequence $(r^1_0, r^1_i), (r^2_0, r^2_i),...,(r^m_0, r^m_i)$ as it does in $(t^1_0, t^1_i), (t^2_0, t^2_i),...,(t^m_0, t^m_i)$
\end{enumerate}
We refer to $m$ as the length of the balanced sequence.

\end{definition}

In the simplest case, a balanced sequence involves two actions $s_0, s'_0$ for the focal player and four pairs of strategy profiles:
$$s_0 \succeq_{s_A, s_B} s'_0, \quad s'_0 \succeq_{s'_A, s_B} s_0, \quad s'_0 \succeq_{s_A, s'_B} s_0, \quad s_0 \succeq_{s'_A, s'_B} s'_0,$$

\noindent in which $\{A,B\}$ is a partition of the opponents, and at least one comparison is strict. Any such sequence is inconsistent with a utility of the form \eqref{eq:addseputility}.  To see this, observe that the four comparisons in turn imply that
$$\sum_{j}g_{0j}(s_0,s_j)\geq \sum_{j}g_{0j}(s_0',s_j),$$
$$\sum_{j\in A} g_{0j}(s_0',s_j')+ \sum_{j\in B}g_{0j}(s_0',s_j)\geq \sum_{j\in A} g_{0j}(s_0,s_j')+ \sum_{j\in B}g_{0j}(s_0,s_j),$$
$$\sum_{j\in A} g_{0j}(s_0',s_j)+ \sum_{j\in B}g_{0j}(s_0',s_j')\geq \sum_{j\in A} g_{0j}(s_0,s_j)+ \sum_{j\in B}g_{0j}(s_0,s_j'), \quad \text{and}$$
$$\sum_{j}g_{0j}(s_0,s_j')\geq \sum_{j}g_{0j}(s_0',s_j').$$

\noindent Since each term appears on the left-hand side as often as it does on the right-hand side, and at least one comparison is strict, adding all of these together gives
\begin{align*}
    \sum_j g_{0j}(s_0,s_j) + &g_{0j}(s_0,s_j') + g_{0j}(s_0',s_j)+g_{0j}(s_0',s_j') \\
     & > \sum_j g_{0j}(s_0,s_j) + g_{0j}(s_0,s_j') + g_{0j}(s_0',s_j) + g_{0j}(s_0',s_j'),
\end{align*}

\noindent a contradiction.  Consequently, no additively separable representation exists. 

Ruling out short balanced sequences exactly captures our intuition about bilateral strategic interactions.  In fact, strategic independence simply rules out particular balanced sequences of length $4$.  Ruling out longer balanced sequences builds on this idea, precluding more complicated strategic dependencies across opponents.  Our next result shows that the absence of balanced sequences of any length is a necessary and sufficient condition for strategic separability on finite action sets.


\begin{Theorem}\label{theo:cyclic}
 Suppose that each $S_i$ is finite.  The utility $u_0$ is strategically separable if and only if its induced family of orders $\{\succeq_{s_+}\}_{s_+ \in S_+}$ does not admit a balanced sequence.
\end{Theorem}

Theorem \ref{theo:cyclic} helps us to interpret earlier findings.  We can always establish strategic separability by ruling out balanced sequences of any length.  Theorem \ref{theo:main2} tells us that, with three completely essential opponents and preferences that satisfy joint solvability, we can get away with only checking short sequences.  The ``only if'' direction follows from a similar argument as in our example above.  Given a balanced sequence, if we had an additively separable utility $u$, property $(c)$ would imply that
$$\sum_{k=1}^m u(r^k) = \sum_{k=1}^m u(t^k),$$

\noindent but property $(b)$ implies
$$\sum_{k=1}^m u(r^k) > \sum_{k=1}^m u(t^k).$$

Of course, many network games following \citet{Ballesteretal2006} have infinite action sets, so Theorem \ref{theo:cyclic} does not apply directly. Example \ref{ex:obstruction} highlights a potential difficulty:  a single scale might lack the richness to express the properties of $\{\succeq_{s_+}\}_{s_+ \in S_+}$. To accommodate this possibility, we enlarge the range of utility functions. Instead of assigning a single number to each strategy-profile, we assign a collection of utilities to each strategy-profile. Two strategies are then compared lexicographically, as in \cite{blume1991lexicographic1, blume1991lexicographic2}.  

Formally, we say that the family of orders $\{\succeq_{s_+}\}_{s_+ \in S_+}$ admits a \textbf{strategically separable lexicographic} representation if there is a linearly ordered index set $M$ and a collection of functions $g_j:S_0\times S_j\rightarrow \mathbb{R}^{M}$ for $j=1,\dots, n$ such that $$s_0\succeq_{s_+} s_0' \iff \sum_{j=1}^n g_j(s_0,s_j) \geq_{lex} \sum_{j=1}^n g_j(s_0',s_j') $$ where $\geq_{lex}$ is the order on $\mathbb{R}^M$ where  $x\geq_{lex} y$ if and only if $x=y$ or $x_l>y_l$, in which $l$ is the minimal coordinate in which $x$ and $y$ differ.\footnote{We can assume that such a minimum exists by taking the range of each $g_j$ to be such that the set of indices with nonzero values forms a well-ordered set.} We can extend Theorem \ref{theo:cyclic} as follows. 

\begin{Theorem}\label{theo:cyclicinfinite}
    A family of orders $\{\succeq_{s_+}\}_{s_+ \in S_+}$ admits a strategically separable lexicographic representation if and only if there are no balanced sequences.
\end{Theorem}

\noindent In the next subsection, we prove the ``if'' direction of Theorem \ref{theo:cyclic}.  In the Appendix, we use logical compactness to derive Theorem \ref{theo:cyclicinfinite} from Theorem \ref{theo:cyclic}.

\subsection{Proof of Theorem \ref{theo:cyclic}}

\hspace{1 pc}
Our proof relies on a close cousin of Farkas' lemma known as Motzkin's Transposition Theorem.  Versions of this result can be found in \citet{border2013alternative} and \citet{StoerWitzgall1970}.

\begin{lemma}[Motzkin's Transposition Theorem]\label{lem:motzkin}
    Let $A, B$ and $C$ be arbitrary rational matrices, each with $d$ columns. Exactly one of the following holds:\footnote{ For two vectors $x$ and $y$, $x\gg y$ means that $x_i>y_i$ for all $i$; $x>y$ means that $x_i\geq y_i$ for all $i$ and $x\neq y$; and $x\geq y$ means that $x_i\geq y_i$ for all $i$.}
\begin{enumerate}
    \item There is a vector $x\in \mathbb{R}^d$ such that 
    \begin{align*}
        &Ax\gg 0 \\        
        &Bx\geq 0 \\        
        &Cx=0
    \end{align*}

    \item There are integer vectors $p>0$, $q\geq 0$ and $r$ such that $$pA + qB + rC=0.$$
    Furthermore, $B$ or $C$ can be omitted.
\end{enumerate}
\end{lemma}




    Assume all action sets are finite.  We show that if $\{\succeq_{s_+}\}_{s_+ \in S_+}$ admits no balanced sequence, then preferences are strategically separable.  Strategic separability holds if we can find a finite collection of real numbers $\{g_{i}(s_0, s_i): i>0, s_0\in S_0, s_i\in S_i\}$ such that
    $$s_0 \succeq_{s_{+}} s'_0 \quad \iff \quad \sum_{i = 1}^ng_{i}(s_0, s_i) \geq \sum_{i = 1}^n g_{i}(s'_0, s_i) \hspace{1cm} \text{ for all }s_0,s_0'\text{ and }s_+.$$
    
    \noindent This is a finite system of linear inequalities indexed by triples $(s_0, s'_0, s_{+}) \in S_0^2 \times S_{+}$, and a solution exists if and only if the preferences are strategically separable.
    
    Let $I$ be the set of triples corresponding to strict preferences, where $s_0 \succ_{s_{+}} s'_0$, and let $J$ be the set of triples corresponding to weak preferences, where $s_0\succeq_{s_{+}} s'_0$ and $s_0 \neq s'_0$. If $I$ is empty, there are no strict preferences, and $\{\succeq_{s_+}\}_{s_+ \in S_+}$ is trivially strategically separable (e.g., taking $g_{i}(s_0, s_i)=0$ for all $i$, $s_0$, and $s_i$).  Now assume $I$ is nonempty, and order the sets $I$ and $J$ arbitrarily. Define a matrix $A$ with one row for each element of $I$ and one column for each triple $(i, t_0, t_i)$ in which $i>0$, $t_0\in S_0$ and $t_i\in S_i$. Set the entry of $A$ at row $(s_0, s'_0, s_{+})$ and column $(i, t_0, t_i)$ to 
    \[\begin{cases}
        1 & \text{ if }s_0=t_0 \text{ and } s_i=t_i \\
        -1& \text{ if }s_0'=t_0\text{ and } s_i=t_i \\
        0 & \text{ otherwise.}
    \end{cases}\]
    Define an analogous matrix $B$ for the weak comparisons, containing one row for each element of $J$.  A vector $\mathbf{g}$ solves our system of inequalities if $A \mathbf{g} \gg 0$ and $B \mathbf{g} \geq 0$.  By Lemma \ref{lem:motzkin}, a solution exists if and only if there are no integer vectors $p>0$ and $q \geq 0$ such $pA+qB=0$.  We show that if such vectors exist, then we can find a balanced sequence.  Hence, if there is no solution, then a balanced sequence exists, completing our proof.
    
    Given integer vectors $p$ and $q$ as above, we can construct a balanced sequence $(r^k, t^k)_k$ as follows. Suppose that $(s_0,s_0', s_{+})$ is the first element of $I$ such that the corresponding entry of $p$ is positive, and suppose this entry equals $\ell$. Set $r^1=r^2=\cdots = r^{\ell}=(s_0,s_{+})$, and $t^1=t^2=\cdots = t^{\ell}=(s'_0,s_{+})$.  Repeat this for all subsequent positive entries of $p$, and then do the same for the positive entries of $q$. Let $|p|$ and $|q|$ denote the sum of the entries in $p$ and $q$ respectively.  This yields a sequence of length $|p| + |q|$, which satisfies conditions (a) and (b) in Definition \ref{def:cyclicsequence} by construction.  Condition (c) is exactly the requirement that $pA+qB=0$. \qedsymbol

\subsection{Discussion}

\hspace{1 pc}
Theorem~\ref{theo:cyclic} provides a complete revealed–preference test for
strategic separability that is valid \emph{without} additional assumptions: a player’s preferences are additively separable exactly when they
admit no balanced sequence. The cost of weak assumptions is that we now must rule out arbitrarily long balanced sequences. Long balanced sequences are both challenging to interpret and difficult to check in practice. Two auxiliary results, proved in the appendix, illuminate how costly that criterion can be in finite games.

\begin{Prop}
\label{prop:cycle_bound_short}
Let $m:=\prod_{i=0}^{n}|S_i|$ be the number of pure–strategy profiles.  If a
finite preference is \emph{not} strategically separable, then it contains a
balanced sequence whose length does not exceed
$L(m)=(m-1)\exp\left(\frac{m-1}{2}\ln(m-1)\right)$.
\end{Prop}

Proposition~\ref{prop:cycle_bound_short} tells us that a balanced sequence, if it exists, must appear within an explicit, though exponential, search horizon. Although one might be able to improve the bound, our next result shows that any bound must grow with the number of strategy profiles.  However far one looks, there are games with preferences whose first offending sequence is even longer.  While universally valid, the balanced sequence test may be exceedingly hard to apply.


\begin{Prop}
\label{prop:no_finite_ax_short}
For every integer $L\ge 1$ there is a finite game $\Gamma$ with a $u_i$ that is
\emph{not} strategically separable and yet contains \emph{no} balanced sequence
of length $\le L$.
\end{Prop}

%


\section{Final Remarks}

\hspace{1 pc}
In addition to uncovering precise conditions behind the classic linear-quadratic utility for network games, our analysis identifies two canonical utilities that relax key assumptions.  If strategic interactions are bilateral, players behave as if their preferences follow an additively separable utility
$$u_i(s_i, s_{-i}) = \sum_{j \neq i} g_{ij}(s_i, s_j).$$

\noindent With real-valued actions, if preferences exhibit a constant rate of substitution between opponents, we can take
$$u_i(s_i, s_{-i}) = b_i(s_i) + \gamma_i(s_i)\sum_{j \neq i} g_{ij} s_j$$

\noindent for some functions $b_i$ and $\gamma_i$, and constants $\{g_{ij}\}_{j \neq i}$.  Results we derive using these functional forms apply to any game in which preferences satisfy these qualitative properties.  Our findings thus clarify implicit assumptions in network game models, and indeed they suggest a natural definition for what it means to be a network game.  Beyond providing a rationale for specific utilities, our results highlight relatively simple revealed preference conditions to test whether such models are appropriate.  Our notion of bilateral strategic interactions exactly corresponds to ruling out a particular class of $4$-cycles in a player's preferences.

By including ``Part I'' in our title, we deliberately call attention to unfinished business.  As our Theorems \ref{theo:CRS} and \ref{theo:linquad} demonstrate, with an additively separable utility in hand, it becomes easier to deduce implications from other qualitative assumptions on preferences.  We expect similar techniques can also shed light on the relationship between network games and potential games.  Conveniently, many network games studied in the literature are also potential games.  This ensures that equilibria exist in pure strategies, and players can realistically learn how to play equilibria from feedback, even if the network is large and complex.  In an ongoing project, which we expect will constitute ``Part II,'' we are working to identify revealed preference conditions implying that a network game has a potential, as well as a canonical form for that potential.

Perhaps a more obvious next task is to derive substantive predictions in network models that do not depend on players having a particular utility function.  Who takes higher or lower actions?  How do shocks to one player's incentives propagate through the network?  Who should we target for interventions?  A popular approach to such questions assumes differentiable utilities and focuses on small shocks, so one can approximate changes using a linear system \citep[e.g.][]{Dasarathaetal2025}.  A complementary approach might leverage our results to better understand responses to larger shocks.

\bibliographystyle{plainnat}
\bibliography{refs}

\appendix

\section{Proof of Theorem \ref{theo:main2}}\label{ap:proofmain}
\hspace{1 pc}
Mirroring the proof of Theorem \ref{theo:main1}, we construct a difference order $\succeq$ on pairs of strategy profiles.  For this construction, we fix a particular $s^*_0 \in S_0$, and for each $s_0 \in S_0$ distinct from $s^*_0$, we apply Lemma \ref{lem:pairwise} to obtain functions $\{g_i\}_{i=1}^n$ such that
$$s_0 \succeq_{s_+} s^*_0 \quad \iff \quad \sum_{i=1}^n g_i(s_0, s_i) \geq 0.$$

\noindent Define $g_i(s^*_0, s_i) = 0$ for all opponents $i$ and actions $s_i \in S_i$.  We now define the difference order in the following steps:
\begin{enumerate}
    \item Define $ss \simeq tt \simeq (s_0^*, s_+)(s_0^*, s'_+)$ for every $s, t \in S$ and every $s_+, s'_+ \in S_+$.

    \item Define $(s_0, s_+)t \succeq (s'_0, s_+) t$ and $t (s'_0, s_+) \succeq t (s_0, s_+)$ whenever $s_0 \succeq_{s_+} s'_0$.

    \item Define $(s_0, s_+)(s_0, t_+) \succeq (s_0, s'_+)(s_0, t'_+)$ whenever
    $$\sum_{i=1}^n g_i(s_0, s_i) - g_i(s_0, t_i) \geq \sum_{i=1}^n g_i(s_0, s'_i) - g_i(s_0, t'_i).$$

    \item For each pair $s_0, s'_0 \in S_0$, with $s_0, s'_0 \neq s^*_0$, choose an opponent profile $s^*_+$ such that $s_0 \simeq_{s^*_+} s'_0 \not\simeq_{s^*_+} s^*_0$, and define
    $$\rho := \frac{\sum_{i=1}^n g_i(s'_0, s_i)}{\sum_{i=1}^n g_i(s_0, s_i)} > 0.$$

    \noindent We now define $(s'_0, s'_+)t \simeq (s_0, \hat{s}_+)t$ and $t(s'_0, s'_+) \simeq t(s_0, \hat{s}_+)$ for all profiles $t$ if
    $$\sum_{i=1}^n g_i(s'_0, s'_i) = \rho \sum_{i=1}^n g_i(s_0, \hat{s}_i).$$

    \item Take the transitive closure.
\end{enumerate}

We show that i) the order $\succeq$ is total on $S^2$ and remains consistent with the preferences $\{\succeq_{s_+}\}_{s_+ \in S_+}$ and with the pairwise utilities $\sum_{i=1}^n g_i(s_0, s_i)$, ii) it is continuous, and iii) it satisfies the hexagon condition.  Given these properties, we can apply Lemma \ref{lem:wakker} to obtain continuous functions $v_1(s)$ and $v_2(s)$ such that
$$st \succeq s' t' \quad \iff \quad v_1(s) + v_2(t) \geq v_1(s') + v_2(t').$$

\noindent As before, part (a) of the construction implies $v_2(s) = -v_1(s) + \lambda$ for some constant $\lambda$, and part (b) implies $v_1$ represents $\{\succeq_{s_+}\}_{s_+ \in S_+}$.  Since the representation is unique up to positive affine transformations, we have
$$v_i(s_0, s_+) = c_1(s_0) \sum_{i=1}^n g_i(s_0, s_i) + c_2(s_0)$$

\noindent for some functions $c_1$ and $c_2$.  Absorbing $c_2$ into one of the terms in the sum yields our representation.

We now turn to establishing properties i), ii), and iii).  We first show that the order is total, meaning we can compare any two pairs $st$ and $s't'$.  Using part (d) of the construction, together with the unbounded influence assumption, we can find $\hat{t}_+$, $\hat{s}'_+$, and $\hat{t}'_+$ so that $st \simeq s (s_0, \hat{t}_+)$ and $s't' \simeq (s_0, \hat{s}'_+)(s_0, \hat{t}'_+)$.  Part (c) now ensures a comparison between $s (s_0, \hat{t}_+)$ and $(s_0, \hat{s}'_+)(s_0, \hat{t}'_+)$, which must carry over by transitivity.
    
We next show that the inclusion of the comparisons (d) do not conflict with the comparisons added in (b) and (c).  This means showing that $s_0 \succeq_{s_+} s'_0$ if and only if
$$\rho\sum_{i=1}^n g_i(s_0, s_i) \geq \sum_{i=1}^n g_i(s'_0, s_i),$$

\noindent in which $\rho$ is the pair specific scaling constant selected in part (d) of the construction---write $s^*_+$ for the corresponding opponent profile at which the sums are equal and $s_0 \simeq_{s^*_+} s'_0$.  If both sums are zero at some $s_+$, we clearly have $s_0 \simeq_{s_+} s^*_0 \simeq_{s_+} s'_0$, and transitivity gives the required comparison.  Without loss, suppose $s_0 \succ_{s_+} s^*_0$, so the sum on the left is positive---the other case is substantively identical.  If the sum on the right is non-positive, we have $s_0 \succ_{s_+} s^*_0 \succeq_{s_+} s'_0$, and transitivity gives the required comparison.

If the sum on the right is also positive, we have more work to do.  We make use of another lemma to address this case.

\begin{lemma}\label{lem:jointsolv}
    Suppose Assumption \ref{as:basic2} holds for action sets $\{S_i\}_{i=0}^n$ and the family of preference orders $\{\succeq_{s_+}\}_{s_+ \in S_+}$.  Take any three actions $s_0, s'_0, s''_0 \in S_0$, and write $\{g_i\}_{i = 1}^n$, $\{g'_i\}_{i=1}^n$, and $\{g''_i\}_{i = 1}^n$ respectively for functions such that
    $$s'_0 \succeq_{s_+} s_0 \, \iff \, \sum_{i=1}^n g_i(s_i) \geq 0, \quad \quad s''_0 \succeq_{s_+} s_0 \, \iff \, \sum_{i=1}^n g'_i(s_i) \geq 0, \text{ and}$$
    $$s''_0 \succeq_{s_+} s'_0 \, \iff \, \sum_{i=1}^n g''_i(s_i) \geq 0.$$

    \noindent Given any two opponents $i, j$, and any actions $s_i, s'_i \in S_i$ and $s_j \in S_j$, we can find $s'_j \in S_j$ so that
    $$g_i(s'_i) - g_i(s_i) = g_j(s'_j) - g_j(s_j), \quad g'_i(s'_i) - g'_i(s_i) = g'_j(s'_j) - g'_j(s_j), \text{ and}$$
    $$g''_i(s'_i) - g''_i(s_i) = g''_j(s'_j) - g''_j(s_j).$$
\end{lemma}

\begin{proof}
Select profiles of actions $s_{-ij}, s'_{-ij}, s''_{-ij}$ for opponents other than $i$ and $j$ such that
$$s_0 \simeq_{s'_i, s_j, s_{-ij}} s'_0, \quad s_0 \simeq_{s'_i, s_j, s'_{-ij}} s''_0, \quad \text{and} \quad s'_0 \simeq_{s'_i, s_j, s''_{-ij}} s_0,$$

\noindent which implies
$$g_i(s'_i) + g_j(s_j) + \sum_{k \neq i,j} g_k(s_k) = g'_i(s'_i) + g'_j(s_j) + \sum_{k \neq i,j} g'_k(s'_k) = g''_i(s'_i) + g''_j(s_j) + \sum_{k \neq i,j} g''_k(s''_k) = 0.$$

\noindent By joint solvability, we can find $s'_j$ so that
$$s_0 \simeq_{s_i, s'_j, s_{-ij}} s'_0, \quad s_0 \simeq_{s_i, s'_j, s'_{-ij}} s''_0, \quad \text{and} \quad s'_0 \simeq_{s_i, s'_j, s''_{-ij}} s_0,$$

\noindent which implies
$$g_i(s_i) + g_j(s'_j) + \sum_{k \neq i,j} g_k(s_k) = g'_i(s_i) + g'_j(s'_j) + \sum_{k \neq i,j} g'_k(s'_k) = g''_i(s_i) + g''_j(s'_j) + \sum_{k \neq i,j} g''_k(s''_k) = 0.$$

\noindent Substituting and canceling common terms gives
$$g_i(s'_i) + g_j(s_j) = g_i(s_i) + g_j(s'_j), \quad g'_i(s'_i) + g'_j(s_j) = g'_i(s_i) + g'_j(s'_j), \quad \text{and}$$
$$g''_i(s'_i) + g''_j(s_j) = g''_i(s_i) + g''_j(s'_j),$$

\noindent which implies the claim.
\end{proof}

Suppose that $s_0 \preceq_{s_+} s'_0$ but
$$u(s_0, s_+) := \rho\sum_{i=1}^n g_i(s_0, s_i) > \sum_{i=1}^n g_i(s'_0, s_i) := u(s'_0, s_+).$$
    
\noindent The analysis is substantively identical if the preference is strict, but the inequality between the sums is weak.  Apply Lemma \ref{lem:pairwise} to obtain functions $\{\tilde{g}_i\}_{i=1}^n$ such that
$$s_0 \succeq s'_0 \quad \iff \quad \sum_{i=1}^n \tilde{g}_i(s_i) \geq 0.$$

\noindent Define $s^0_+ = s^*_+$ and $s^1_+ = s_+$.  Since the set
$$A^1 := \{t_+ \, : \, 0 = \sum_{i=1}^n \tilde{g}_i(s^0_i) \geq \sum_{i=1}^n \tilde{g}_i(t_i) \geq \sum_{i=1}^n \tilde{g}_i(s^1_i)\}$$
    
\noindent is connected (this follows from the connected intervals assumption), and the function $u$ is continuous for each fixed $s_0$, given any positive integer $N$ we can find a point $s_+^{1/N} \in A^1$ such that
$$u(s_0, s^{1/N}_+) = \frac{N-1}{N}u(s_0, s^0_+) + \frac{1}{N}u(s_0, s^1_+).$$

\noindent We can then similarly find $s^{2/N}_+ \in A^2 := \{s_+ \, : \, \sum_{i=1}^n \tilde{g}_i(s^{1/N}_i) \geq \sum_{i=1}^n \tilde{g}_i(s_i) \geq \sum_{i=1}^n \tilde{g}_i(s^1_i)\}$ such that
$$u(s_0, s^{2/N}_+) = \frac{N-2}{N}u(s_0, s^0_+) + \frac{2}{N}u(s_0, s^1_+).$$

\noindent Continuing in this way, we construct a sequence of opponent profiles $s^{k/N}_+$ with
$$u(s_0, s^{k/N}_+) = \frac{N-k}{N}u(s_0, s^0_+) + \frac{k}{N}u(s_0, s^1_+),$$

\noindent and $\sum_{i=1}^n \tilde{g}_i(s^{(k+1)/N}_i) \leq \sum_{i=1}^n \tilde{g}_i(s^{k/N}_i)$ for each $k$.
    
There must exist an increment $s_+^{k^*/N}$ to $s_+^{(k^*+1)/N}$ such that
$$u(s'_0, s^{(k^*+1)/N}_+) - u(s'_0, s^{k^*/N}_+) \leq \frac{1}{N}\left(u(s'_0, s^{1}_+) - u(s'_0, s^{0}_+)\right).$$

\noindent Through repeated application of Lemma \ref{lem:jointsolv}, we can copy this increment and ``subtract'' it from $s_+^0$, obtaining a profile $s^{-1}_+$ with
$$u(s_0, s^{0}_+) - u(s_0, s^{-1}_+) = \frac{1}{N}\left(u(s_0, s^{1}_+) - u(s_0, s^{0}_+)\right),$$
$$u(s'_0, s^{0}_+) - u(s'_0, s^{-1}_+) \leq \frac{1}{N}\left(u(s'_0, s^{1}_+) - u(s'_0, s^{0}_+)\right), \quad\text{and}\quad \sum_{i=1}^n \tilde{g}_i(s^{-1}_i) \geq 0.$$

\noindent We can then analogously construct a sequence of profiles $s^{-2}_+$, $s^{-3}_+$,...  If we chose a large enough $N$ at the start of this procedure, we eventually reach a profile $s^{-K}_+$ at which
$$u(s_0, s^{-K}_+) < 0, \quad u(s'_0, s^{-K}_+) \geq 0, \quad \text{and} \quad \sum_{i=1}^n \tilde{g}_i(s^{-K}_i) \geq 0,$$

\noindent which implies $s_0 \prec_{s^{-K}_+} s^*_0 \preceq_{s^{-K}_+} s'_0 \preceq_{s^{-K}_+} s_0$, a contradiction.  We conclude that $s_0 \succeq s'_0$ whenever
$$\rho \sum_{i=1}^n g_i(s_0, s_i) \geq \sum_{i=1}^n g_i(s'_0, s_i).$$

Addressing continuity, we show that the order $\succeq_0$ on $S$ defined by
$$s \succeq_0 s' \quad \iff \quad st \succeq s't$$

\noindent for some profile $t$ is continuous.  Continuity of the difference order then follows from an argument analogous to that in Theorem \ref{theo:main1}.  We show that the upper contour sets $\{t \, : \, t \succ_0 s\}$ are open---analogously, the lower contour sets $\{t \, : \, t \prec_0 s\}$ are open, implying that $\succeq_0$ is continuous.
    
Suppose $t \succ_0 s$.  Since $\succeq_{t_+}$ is a continuous order on $S_0$, we can choose $t^*_0$ so that $t \succ_0 (t^*_0, t_+) \succ_0 s$.  Let $\rho$ be the scaling factor for $t^*_0$ and $s_0$ from part (d) in the construction.  Since $(t^*_0, t_+) \succ_0 s$, we have
$$\sum_{i=1}^n g_i(t^*_0, t_i) > \rho\sum_{i=1}^n g_i(s_0, s_i).$$

\noindent Since the functions $g_i$ are continuous if we fix $t^*_0$, there is an open neighborhood $V_+$ of $t_+$ such that
$$\sum_{i=1}^n g_i(t^*_0, r_i) > \rho\sum_{i=1}^n g_i(s_0, s_i)$$

\noindent whenever $r_+ \in V_+$, which implies $(t^*_0, r_+) \succ_0 s$ whenever $r_+ \in V_+$.  By joint continuity, there exist open neighborhoods $U_0$ of $t_0$ and $U_+$ of $t_+$ such that $r_0 \succ_{r_+} t^*_0$ whenever $r_0 \in U_0$ and $r_+ \in U_+$.  This implies
    $$(r_0, r_+) \succ_0 (t^*_0, r_+) \succ_0 s$$

\noindent whenever $r_+ \in U_+ \cap V_+$.  Since the intersection of two open sets is open, the set $U_0 \times (U_+ \cap V_+)$ is an open neighborhood of $t$ such that $r \succ_0 s$ for each $r$ in the neighborhood.  We conclude that the upper contour set is open.

Finally, the argument that $\succeq$ satisfies the hexagon condition is analogous to that in the proof of Theorem \ref{theo:main1}.  \qedsymbol

\section{Proof of Results From Section \ref{sec:cyclic}}

\subsection{Proof of Theorem \ref{theo:cyclicinfinite}} \label{sec: compactness extension}

Recall that Theorem \ref{theo:cyclicinfinite} says a family of orders $\{\succeq_{s_+}\}_{s_+ \in S_+}$ does not admit a balanced sequence if and only if admits a strategically separable lexicographic representation. We prove this in two main steps. First we, use logical compactness together with the result of Theorem \ref{theo:cyclic} to deduce that a strategically separable representation exists, with utilities taking values in some arbitrary linearly ordered Abelian group. We then use Hahn's embedding theorem to embed this group in to $\mathbb{R}^M$ for some linearly ordered index set $M$. Several papers have outlined how to use logical compactness to bootstrap finitary models to arbitrary models, including \citet{FuhrkenRichter1991} and \citet{gonczarowski2025infinity}. We take a similar approach. Logical compactness applies in first-order logic, where one cannot axiomatize the real numbers. As a consequence, we cannot use compactness to generate a real-valued utility. However, linearly ordered Abelian groups can be finitely axiomatized so we can use logical compactness to get a representation in such a set. A common approach is then to impose an Archimedean axiom. Since Holder's theorem characterizes the reals as the unique Archimedean liniearly ordered group, one can indirectly generate a real-valued representation. Instead, we do not impose such a criterion, but allow for non-Archimedean orders. As a consequence, our representation takes values in a richer space, namely $\mathbb{R}^M$.

Neither first-order logic nor the theory of ordered Abelian groups is commonly used in economic theory, so we give a brief primer on both for the reader's convenience.\footnote{For more details on ordered Abelian groups see \citet{fuchs2011partially}. For more on mathematical logic see \citet{enderton2001mathematical}.}




\subsubsection{Background: First-Order Logic}
A \emph{language} $L$ is a list of symbols. It includes logical connectives ($\land,\lor,\to,\neg$), quantifiers ($\forall,\exists$), and parentheses $(,)$. It can also include
\begin{itemize}
  \item A collection of constant symbols $c_0,c_1,\dots$
  \item A collection of function symbols $f_0,f_1,\dots$, each with a specified arity $n_0,n_1,\dots$
  \item A collection of predicate (or relation) symbols $R_0,R_1,\dots$, each with a specified arity $m_0,m_1,\dots$
\end{itemize}
At this point, the arity of a function or relation is simply a positive integer attached to each of the associated symbols. The arity represents the number of arguments it takes in. For example, the function $f(x,y)=x^2y$ has arity $2$. The arity of a relation is likewise the number of arguments it allows. For example, a binary relation $R$ on a set $X$ is a subset of $X\times X$, so has arity $2$. 

Formulas are strings of symbols built from the language inductively. We omit the rules for building formulas for brevity. The rules imply that  $\forall u,v,w\ (u>v)\to (u+w>v+w)$ is a valid formula in a language that includes a binary relation $>$ and a binary function $+$. However, the string $\forall u+\to >$ is not a valid formula. The formula $\forall u (u>v)$ is valid but includes the variable $v$ which is not bound to any quantifier. A formula where all variables are bound to a quantifier is called a $L$-sentence. The first example above is a sentence, for example. A \emph{theory} $\mathcal{T}$ is a set of $L$-sentences. Note that while each sentence includes only a finite number of symbols, the number of sentences in a theory is not required to be finite.

$L$ is simply a list of symbols. Sentences written in $L$ get meaning when all the symbols are given an interpretation in a specific setting. An \emph{$L$-structure} (or \emph{model}) interprets each symbol in a given set $X$, called the domain. Constants are interpreted as elements of the domain, functions as operations, and predicates as relations.   A structure satisfies (or is a \emph{model} of) $\mathcal{T}$ if it makes every sentence in $\mathcal{T}$ true.

A fundamental result in first-order logic is the \emph{Compactness Theorem}:
\begin{Theorem}[Compactness Theorem]
    If every finite list of sentences of a first-order theory $\mathcal{T}$ has a model, then the entire theory $\mathcal{T}$ has a model.  
\end{Theorem}

The compactness theorem is our main tool in the section. There is, however, a drawback to working in first-order logic: there is no theory in first-order logic that uniquely axiomatizes the real numbers. This is because we are not allowed to quantify over sets, so one cannot express, for example, the least upper bound property. We will circumvent this issue by simply requiring a representation in an ordered Abelian group, which we can axiomatize. 

\subsubsection{Background: Ordered Abelian Groups}

Recall that an \textit{Abelian group} is a set $G$ together with a binary operation $+$ such that $+$ is commutative and associative. $G$ is required to have an additive identity, usually denoted $0$ and additive inverses so that for every $g\in G$ there is a $h\in G$ with $g+h=0$. We generally write $h$ as $-g$.

A \textit{linearly ordered Abelian group} is a triple $(G,+,\succeq)$ such that $(G,+)$ is an Abelian group and $\succeq$ is a linear order on $G$ that is translation-invariant so that $g\succeq h$ implies $g + z\succeq h +z$ for all $g,h,z$. An \textit{isomorphism} between two linearly ordered groups is a bijection that preserves order (so is ``isotone") and preserves the binary operation (so is a ``homomorphism").

For any $g$ in a linearly ordered Abelian group $(G,+,\succeq)$, let $\vert g \vert$ be the larger of $g$ and $-g$. $g$ and $h$ are said to be \emph{Archimedean equivalent} if there are natural numbers $n$ and $m$ such that $n\vert g \vert \geq \vert h \vert$ and $m\vert h \vert \geq \vert g \vert$. $(G,+,\succeq)$ is said to be Archimedean if any two elements are Archimedean equivalent. It's not hard to show that the equivalence classes of $G$ can be ordered according to the ordering of their members. 

Hahn's embedding theorem characterizes all linearly ordered Abelian groups as follows. For any linearly ordered set $M$, let $\mathbb{R}^{[M]}$ denote the set of all functions $f:M\rightarrow \mathbb{R}$ such that $\{m\,:\, f(m)\neq 0\}$ is well-ordered. This is a group with component-wise addition. It is a linearly ordered group with the lexicographic order, given by $f\geq_{lex}g$ if $f=g$ or $f\neq g$ and $f(m)>g(m)$ for the smallest $m\in M$ where $f$ and $g$ disagree. 

\begin{Theorem}[Hahn's Embedding Theorem]
    For every linearly ordered Abelian group $(G,+,\succeq)$, there is a linearly ordered set $M$ such that $(G,+,\succeq)$ is isomorphic to a subgroup of $\mathbb{R}^{[M]}$. Furthermore, $M$ can be taken to be the set of Archimedean equivalence classes of $G$ with the natural order.
\end{Theorem}





\subsubsection{First-Order Language for Preferences}
We now construct a specific theory for our purposes. Fix a strategy set $S$ and a set of agents $N=\{0,1,\dots, n\}$. Fix a family of  preference relations $\{\succeq_{s_+}\}_{s_+ \in S_+}$ that admits no balanced sequence (for notational simplicity, we assume that $S_i=S$ for all $i$, although the general case is not more difficult). Let $L$ be the language consisting of:
\begin{enumerate}
  \item Constant symbols $c^i_{(s,t)}$ for each pair $(s,t)\in S\times S$ and each $i=1,2,\dots, n$
  \item The constant symbol $0$
  \item A binary function symbol $+$
  \item A binary relation symbol $>$.
\end{enumerate}

\subsubsection{Theory $\mathcal{T}$}
We now construct a theory $\mathcal{T}$ in language $L$ that contains:

\paragraph{Ordered abelian group axioms:}
\begin{align*}
  &\forall u,v,w\hspace{.3cm} (u+v)+w = u+(v+w),\\
  &\forall u,v\hspace{.3cm} u+v = v+u,\\
  &\forall u\hspace{.3cm} u+0 = u,\quad \forall u\ \exists v\ u+v=0,\\
  &\forall u,v,w\hspace{.3cm} (u>v \wedge v>w)\to u>w,\\
  &\forall u,v\hspace{.3cm} (u>v)\lor (u=v)\lor (v>u),\\
  &\forall u,v,w\hspace{.3cm} (u>v)\to (u+w>v+w),\\
  &\neg(0>0).
\end{align*}

\paragraph{Revealed-preference axioms:}
For every pair of profiles $(s_0,s_1\dots,s_n),(s_0',s_1,\dots,s_n)\in S^n$, such that
$s_0\succ_{s_+}s_0'$ include the axiom
\[
  c^1_{(s_0,s_1)}+c^2_{(s_0,s_2)}+\cdots+c^n_{(s_0,s_n)}>  c^1_{(s_0',s_1)}+c^2_{(s_0',s_2)}+\cdots+c^n_{(s_0',s_n)}
\]Likewise, for every pair of profiles $(s_0,s_1\dots,s_n),(s_0',s_1,\dots,s_n)\in S^n$, such that
$s_0\sim_{s_+}s_0'$ include the axiom
\[
  c^1_{(s_0,s_1)}+c^2_{(s_0,s_2)}+\cdots+c^n_{(s_0,s_n)}=  c^1_{(s_0',s_1)}+c^2_{(s_0',s_2)}+\cdots+c^n_{(s_0',s_n)}
\]

\subsubsection{Application of Logical Compactness}
If $\{\succeq_{s_+}\}_{s_+ \in S_+}$ admits no balanced sequence, then the same family restricted to any finite subset of $S^n$ admits a real-valued, additively separable utility representation by Theorem \ref{theo:cyclic}. 
Any finite subset of $\mathcal{T}$ mentions only finitely many constants. The absence of balanced sequences on this set ensures an additively separable representation by Theorem \ref{theo:cyclic}. 
By the Compactness Theorem, $\mathcal{T}$ is therefore satisfiable. Let $(G,+,\geq)$ be the associated linearly ordered Abelian group. Define $g_i(s_0,s_i)=c^i_{(s_0,s_i)}$ for each $i$ and each $(s_o,s_i)$. The revealed-preference axioms imply that $s_0\succsim_{s_+}s_0'$ if and only if $\sum_{i=1}^n g_i(s_0, s_i) \geq \sum_{i=1}^n g_i(s_0', s_i)$.

Applying Hahn's embedding theorem, there is a linearly ordered set $M$, a subgroup $H$ of $\mathbb{R}^{[M]}$ and an isomorphism $\Phi$ between $(G,+,\geq)$ and $(H, +, \geq_{lex})$. Setting $\hat{g}_i =\Phi\circ g_i$ for each $i$ gives the desired represntation.

\subsection{Proof of Proposition \ref{prop:cycle_bound_short}}\label{sec: cycle bounds}

In this section, we provide a bound on the length of balanced sequences that need to be checked, fixing the cardinalities of the action sets. A similar approach was taken to bound the size of solutions to linear programs in \cite{echenique2022efficiency}. Our main tool will be a lemma that we prove for general homogeneous linear systems which may be of interest outside of network games.

For any integer $k$ and real matrix $M$, let
\[
  \lambda_k(M)
  =
  \max\bigl\{\,|\det N|\colon N\ \text{is a}\ k\times k\ \text{submatrix of }M\bigr\},
\]
and then set
\[
  \Lambda_k(M)
  =
  \max_{0\leq \ell\le k}\,\lambda_\ell(M),
\]
where by convention the determinant of a $0\times0$ matrix is $1$.  For a scalar $c$, we write $\mathbf{c}_k\in\mathbb{R}^k$ for the constant vector whose entries are all $c$, and $I_k$ for the $k\times k$ identity. Given a matrix $A$, we write $A_l$ for the $l$th row of $A$.


\begin{lemma}\label{lem:lin prog}
  Let 
  \[
    A\in\mathbb{Z}^{a\times d},\quad
    B\in\mathbb{Z}^{b\times d},\quad
    C\in\mathbb{Z}^{c\times d},
  \]
  and form
$$M=\left(\begin{array}{c}
         A  \\
         B  \\
         C
    \end{array}\right)$$
  Suppose there exists \(x\in\mathbb{R}^d\) satisfying
  \[
    A x \,\gg\, \mathbf{0}_a,\quad
    B x \,\ge\,\mathbf{0}_b,\quad
    C x \,=\,\mathbf{0}_c.
    \tag{I}
  \]
  Then there is an \emph{integral} solution \(z\in\mathbb{Z}^d\) to the same system such that
    $|z_i|\le a\;\Lambda_{d-1}(M)$ for each $i$.
\end{lemma}

\begin{proof}
  Consider instead the system
  \[
    R\,x \;\ge\; b,
    \tag{II}
  \]
  where
  \[
    R \;=\;
    \begin{pmatrix}
      A\\[4pt]
      B\\[4pt]
      C\\[4pt]
      -C\\[4pt]
      -I_n\\[2pt]
      I_n
    \end{pmatrix}
    \in\mathbb{Z}^{(a+b+2c+2d)\times d},
    \qquad
    b =
    \begin{pmatrix}
      \mathbf{0}_a\\[3pt]
      \mathbf{0}_b\\[3pt]
      \mathbf{0}_c\\[3pt]
      \mathbf{0}_c\\[3pt]
      -\mathbf{1}_n\\[2pt]
      -\mathbf{1}_n
    \end{pmatrix}.
  \]
  $Rx\geq b$ includes all the previous inequalities from $(I)$, but also bounds any solution to be in the unit cube $[-1,1]^d$. Any solution $x$ of $(I)$ can be scaled by \(1/\max_i|x_i|\) to satisfy $(II)$ while retaining the homogeneous inequalities in $(I)$. 
  
  Now
   \[ P \;=\;\{\,x\in\mathbb{R}^d : R\,x\ge b\}
  \]
  is a non-empty and bounded polytope, and for each  strict inequality from $A$, indexed by $\ell$ with $1\leq \ell \leq a$, there is a vertex $z^{\ell}$ of \(P\) at which the strict inequality \(A_\ell\,z^{\ell}>0\) holds. As a vertex, $z^{\ell}$ is the unique solution to \(d\) linearly independent binding constraints, that is $R'z^{\ell}=b'$ where $R'$ is obtained by dropping all but $d$ linearly independent rows of $R$ and $b'$ is the corresponding subvector of $b$. Notice that at least one of the cube inequalities must be included in the rows of $R'$, since otherwise $\mathbf{0}_d$ would also be a solution, contradicting uniqueness ($z^{\ell}$ is not zero since by assumption \(A_\ell\,z^{\ell}>0\)). The coordinates $i$ where $z^{\ell}_i=1$ and $z^{\ell}_i=-1$ can be reduced from the system to get a smaller family of linearly independent equalities. Call the corresponding submatrix \(R^{\ell}\) and right-hand side \(b^{\ell}\), so
  \[
    R^{\ell}\,\tilde z^{\ell} \;=\; b^{\ell}.
  \]
  where $\tilde z^{\ell}$ is the subvector of $ z^{\ell}$ dropping those coordinates $i$ where $z^{\ell}_i=1$ and $z^{\ell}_i=-1$. Note that $R^{\ell}$ includes only rows from $A,B,C$ and $-C$.
  
  Cramer’s rule then implies that each coordinate of \(\tilde z^{\ell}\) (and thus $z^{\ell}$) is a rational number with denominator \(\det R^{\ell}\), and since \(z^{\ell}\in P\) we know \(\bigl|z^{\ell}_i\bigr|\le1\).  Then
  \[
    \vert\det(R^{\ell})\vert\,z^{\ell} \;\in\;\mathbb{Z}^d
    \quad\text{and}\quad
    \vert\det(R^{\ell})\vert\,z^{\ell}_i\;\le\;\vert\det(R^{\ell})\vert \text{ for all }i .
  \]
  Since $R^{\ell}$ includes at most $d-1$ rows from $A$, $B$, $C$ and $-C$, we have that
  \(\bigl|\det(R^{\ell})\bigr|\le\Lambda_{d-1}(M).\)
  
  \noindent Finally, set
  \[
    z \;=\;\sum_{\ell=1}^a \vert\det(R^{\ell})\vert\,z^{\ell}.
  \]
  where the same construction is used to get $z^{\ell}$ for different values of $\ell\in 1,2, \dots, a$.
  Then \(z\in\mathbb{Z}^d\) and by construction
  \(A z\gg\mathbf0_a,\;Bz\ge\mathbf0_b,\;Cz=\mathbf0_c\).  Since each coordinate of
  \(\vert \det(R^{\ell})\vert\,z^{\ell
  }\) has absolute value at most \(\Lambda_{d-1}(M)\), we obtain
  \[
    |z_i|
    \;\le\;
    \sum_{\ell=1}^a
    \bigl|\det(R^{\ell})\bigr|
    \;\le\;
    a\,\Lambda_{d-1}(M),
  \]
  as claimed.
\end{proof}

\begin{Cor}\label{cor: dual bound}
    Suppose that all entries of $A,B$ and $C$ are $0,1$ or $-1$, then there is an \emph{integral} solution \(z\in\mathbb{Z}^d\) to $(I)$ with 
  \[
    |z_i|\;\le\;\exp\left(\ln(a)+\frac{d-1}{2}\ln(d-1)\right)\;
    \quad\text{for each }i.
  \]
\end{Cor}

\begin{proof} 
The rows of a $k\times k$ submatrix of $M$ have norm at most $\sqrt k$, so by Hadamard's inequality the determinant of a $k\times k$ submatrix is at most $\sqrt k^k$. Plugging this bound into Lemma \ref{lem:lin prog}, and changing form using exponentiation and the log gives the result. 
\end{proof}

Recall that $m:=\prod_{i=0}^{n}|S_i|$. Using transitivity, we can reduce some redundant inequalities in the system of linear inequalities constructed proof of Theorem \ref{theo:cyclic} so that $\vert I \vert + \vert J \vert =\vert S_0-1\vert \cdot \vert S_+ \vert < m$. Then the corresponding dual system is the solution to 

\begin{align*}
p_k &\ge 0\;\forall k \\
q_l &\ge 0\;\forall l \\
p \cdot \mathbf{1} &> 0 \\
pA + qB &= 0
\end{align*}

\noindent where the third inequality ensures that $p_k>0$ for some $k$. 

Applying Corollary \ref{cor: dual bound} to the dual system derived in the proof of Theorem \ref{theo:cyclic} for strategic separability in the finite case, $a=1$ since there is only one strict inequality. $d$ is the total number of constraints in the system which as already noted is less than $m$. Plugging these into Corollary \ref{cor: dual bound}, we can find a dual solution $z$ where   \[
    |z_i|\;\le\;\exp\left(\frac{m-1}{2}\ln(m-1)\right)\;
    \quad\text{for each }i.
  \]

\noindent The associated balanced sequence therefore has length at most $(m-1)\exp\left(\frac{m-1}{2}\ln(m-1)\right)$.

\subsection{Proof of Proposition \ref{prop:no_finite_ax_short}}\label{sec: impossibility of finite axioms}

In this section, we establish the impossibility of a finite first-order axiomatization of strategic separability. We do this by reducing a similar result from the theory of additive separability to our setting. We show that any finite axiomatization for strategic separability could be used to obtain a finite axiomatization for additive separability in the theory of choice. This was shown to be impossible by \citet{luce2007foundations}. Therefore, no finite axiomatization of strategic separability is possible. We begin by formalizing our setting in first-order logic. Rather than encode preferences as a collection of binary relations $\{\succeq_{i,s_{-i}}\}_{s_{-i}\in S_{-i}}$ we can equivalently write them through a single $n+2$-ary relation $R$ defined so that $$(a,b,s_{-i})\in R\iff a\succ_{i,s_{-i}} b.$$ This makes possible a simpler expression of our setting in first-order logic, but it introduces two complications. First, a $n+2$-ary relation on a set $X$ is a subset of $X^{n+2}$. However, in our case the $s_i$ belong to potentially different sets. We will overcome this issue by letting $X=\cup S_i$ and insisting that an element of $X^{n+2}$ is in $R$ only if its coordinates are in the correct subsets. Second, we need to write strategic separability in terms of $R$. We say that $R$ is \textbf{strategically separable} if there are functions 
$g_{ij}$ such that $$(a,b,s_{-i})
\in R \iff \sum_{j\neq i} g_{ij}(a,s_j) \geq \sum_{j\neq i} g_{ij}(b,s_j).$$ Going forward, we will use the notation $(a,b,s_{-i})\in R$ as interchangeable with $R(a,b,s_{-i})$. 

We fix $n$ and index the players $0,1,2,\dots, n$. We will be concerned with the utility of the player $0$. We have a vocabulary consisting of $n+1$ unary relations $S_i$ for $i=0,1,\dots, n$ and one $n+2$-ary relation $R$. The $S_i$ are used to describe which actions belong to which player and $R$ is player $0$'s preferences over their own actions, conditional on the actions of the other players. Intuitively, we interpret $R(a,b,s_1,\dots, s_n)$ to mean that $0$ prefers to play action $a$ to action $b$ when the other players are playing $s_1,\dots, s_n$. We need to encode the requirement that $0$'s preference is complete and transitive over their own actions for each profile of actions over the other players. The following axioms do exactly that. 

\begin{enumerate}
    \item[(A1)] $\forall s \left(S_0(s)\lor S_1(s)\lor S_2(s) \lor \cdots \lor S_n(s)\right)$  
    \item[(A2)] $\forall a,b,s_1,\dots, s_n \left[ R(a,b,s_1,\dots, s_n)\implies \left(S_0(a)\land S_0(b)\land S_1(s_1)\land \cdots \land S_n(s_n)\right)\right]$
    \item[(A3)] $\forall a,b,s_1,\dots, s_n \left[\left(S_0(a)\land S_0(b)\land S_1(s_1)\cdots \land S_n(s_n)\right)\right.$ \\
    $\implies \left.\left(R(a,b,s_1,\dots, s_n) \lor R(a,b,s_1,\dots, s_n)\right)\right]$
    \item[(A4)] $\forall a,b,c,s_1,\dots, s_n \left[\left(R(a,b,s_1,\dots, s_n) \land R(b,c,s_1,\dots, s_n)\right)\implies R(a,c,s_1,\dots, s_n)\right]$
\end{enumerate}

We interpret $S_i(s)$ to mean that $s$ is one of the actions available to agent $i$. Then axiom $(A1)$ simply says that all actions under consideration are available to some agent. Any model of these axioms corresponds exactly to a collection
$\{\succeq_{i,s_{-i}}\}_{s_{-i}\in S_{-i}}$ where we reindex so that $i$ zero and set $$R(a,b,s_1,\dots, s_n)\iff a\succ_{0,(s_1,\dots, s_n)} b.$$ Axiom $(A1)$ and $(A2)$ simply serve to record which agents can take which actions and that preferences are only defined over action profiles where agents take their actions. $(A3)$ is completeness and $(A4)$ is transitivity. 

Suppose that we could add an axiom $A^*$ to this list such that the finite models satisfying $(A1)-(A4)$ and $A^*$ are exactly those that are strategically separable (we can always concatenate a finite list of axioms into a single axiom). We will show that this would imply that we can also add finitely many additional axioms to characterize additive separability in decision theory. \cite{luce2007foundations} established the impossibility of a finite axiomatization of additive separability. Thus finding such an $A^*$ is impossible. 

Suppose now that $n=2m\geq 4$. Consider the following additional axioms

\begin{enumerate}
    \item[(A5)] $\exists x,y\left(x\neq y\right)\land \left(\forall s \left(S_0(s)\implies (s=x)\lor (s=y)\right)\right)$
    \item[(A6)]$\forall s \left(S_1(s)\iff S_2(s)\iff S_3(s)\iff \cdots \iff S_n(s)\right)$
    \item[(A7)] $\forall a,b\forall s_1,\dots s_m \forall s_{m+1},\dots s_n \forall t_{m+1},\dots, t_n$ 
    \begin{align*}
        R(a,b,s_1,\dots, s_m, s_{m+1},\dots, s_n)\land R(a,b,s_{m+1},\dots, s_n,t_{m+1},\dots, t_n)& \\
        \implies R(a,b,s_{1},\dots&, s_m,t_{m+1},\dots, t_n)
    \end{align*}
\end{enumerate}

We now argue that $(A1)-(A7)$ and $A^*$ together axiomatize additive separability in decision theory. To show this, we will establish that any model of $(A1)-(A7)$ and $A^*$ gives rise to an additively separable binary relation and conversely any additively separable total order can be written as a model of $(A1)-(A7)$ and $A^*$.

Suppose that $X$ is a finite set of at least two elements and that $\succeq$ is a total order on $X^m$ that is additively separable so that we can find functions $u_i:X\rightarrow \mathbb{R}$ for $i=1,\dots, m$ where $$(x_1,\dots, x_m)\succeq (y_1,\dots, y_m) \iff \sum_i u_i(x_i)\geq \sum_i u_i(y_i).$$ We now construct a model of $(A1)-(A7)$. Let $X$ be the domain. Pick two distinct elements $z_0$ and $z_1$ of $X$. Let $S'_0$ be the unary relation consisting of $z_0$ and $z_1$ only. Let $S_j'=X$ for all $j>0$. Let $R'$ be the unique $n+2$-ary relation on $\{z_0,z_1\}^2\times X^{2m}$ such that 

\begin{enumerate}
    \item $R'(z_1,z_1, x_1,\dots, x_m, y_1,\dots, y_m)$ and $R'(z_0,z_0, x_1,\dots, x_m, y_1,\dots, y_m)$
    \item $R'(z_1,z_0, x_1,\dots, x_m, y_1,\dots, y_m)\iff (x_1,\dots, x_m)\succeq (y_1,\dots, y_m)$
    \item $R'(z_0,z_1, x_1,\dots, x_m, y_1,\dots, y_m)\iff (y_1,\dots, y_m)\succsim (x_1,\dots, x_m)$
\end{enumerate}

This model satisfies axioms $(A1)-(A7)$ by construction. It remains to show that it satisfies $A^*$. If we can show that $R'$ is strategically separable, we are done since, by assumption, $(A1)-(A4)$ and $A^*$ axiomatize strategic separability. To this end, let $$g_{0j}(z,w)=
\begin{cases}
    0 & \text{ if }z=z_0 \\
    u_j(w)& \text{ if }j\in \{1,\dots, m\} \\
    -u_{j-m}(w)& \text{ if }j\in \{m+1,\dots, n\} 
\end{cases}$$

 Now, \begin{align*}
    R'(z_1,z_0, x_1,\dots, x_m, y_1,\dots, y_m)&\\ 
    \iff &(x_1,\dots, x_m)\succeq (y_1,\dots, y_m)  \\
    \iff &\sum_i u_i(x_i)\geq \sum_i u_i(y_i) \\
    \iff &\sum_i u_i(x_i)+\sum_i -u_i(y_i) \geq 0 \\ 
    \iff \sum_{i=1}^m g_{0i}(z_1,x_i)& + \sum_{i=m+1}^n g_{0i}(z_1,y_i) \geq \sum_{i=1}^m g_{0i}(z_0,x_i) + \sum_{i=m+1}^n g_{0i}(z_0,y_i)
\end{align*}

Similarly, we have 
\begin{align*}
    R'(z_0,z_1, x_1,\dots, x_m, y_1,\dots, y_m)&\\ 
    \iff &(y_1,\dots, y_m)\succeq (x_1,\dots, x_m)  \\
    \iff &\sum_i u_i(y_i)\geq \sum_i u_i(x_i) \\
    \iff &\sum_i u_i(x_i)+\sum_i -u_i(y_i) \leq 0 \\ 
    \iff \sum_{i=1}^m g_{0i}(z_1,x_i)& + \sum_{i=m+1}^n g_{0i}(z_1,y_i) \leq \sum_{i=1}^m g_{0i}(z_0,x_i) + \sum_{i=m+1}^n g_{0i}(z_0,y_i)
\end{align*}

so that $R'$ is strategically separable. 

Next, suppose that $R'$ is a model (with domain $X$) of axioms $(A1)-(A7)$ and $A^*$. Let $z_0$ and $z_1$ be the two elements of $S(0)$. Define $\succeq$ on $X^{m}$ by $$(x_1,\dots, x_m)\succeq (y_1,\dots, y_m) \iff R'(z_1,z_0,x_1,\dots, x_m, y_1,\dots, y_m).$$  By assumption, $R'$ is strategically separable so we can find  $g_{0j}$ such that $$R'(a,b, x_1,\dots, x_m, y_1,\dots, y_m) \iff \sum_{i=1}^m g_{0i}(a,x_i) - \sum_{i=m+1}^n g_{0i}(a,y_i) \geq \sum_{i=1}^m g_{0i}(b,x_i) - \sum_{i=m+1}^n g_{0i}(b,y_i)$$ (the minus signs help with the next step and do not change the results)  Rewriting the term on the right, we have 
$$\sum_{i=1}^m \left(g_{0i}(a,x_i) - g_{0i}(b,x_i)\right) \geq \sum_{i=m+1}^n \left(g_{0i}(a,y_i) -  g_{0i}(b,y_i)\right).$$ Let $h_i(x)=g_{0i}(a,x)-g_{0i}(b,x)$ for $i=1,\dots, m$ and let $r_j(y)=g_{0m+j}(a,y)- g_{0j+j}(b,y)$ for $j=1,\dots,m.$ Then this becomes $$\sum_{i=1}^m h_i(x_i) \geq \sum_{i=1}^m r_i(y_i).$$ Putting this together, we have $$(x_1,\dots, x_m)\succeq (y_1,\dots, y_m) \iff \sum_{i=1}^m h_i(x_i) \geq \sum_{i=1}^m r_i(y_i).$$ Unfortunately, this is not yet an additvely separable representation since $h_i\neq r_i$. However, this representation guarantees the existence of an additively separable representation. To see this, suppose we have 
\begin{align*}
    (x_1^1,\dots, x_m^1) &\succeq (y_1^1,\dots, y_m^1) \\
    (x_1^2,\dots, x_m^2) &\succeq (y_1^2,\dots, y_m^2) \\
    \vdots&\hspace{1cm} \vdots \\
    (x_1^q,\dots, x_m^q) &\succeq(y_1^q,\dots, y_m^q) \\
\end{align*}
Where for each $l$ the multiset $\{x_l^1,\dots, x_l^q\}$ is the same as the multiset $\{y_l^1,\dots, y_l^q\}$. The existence of the $h_i$ and $r_i$ above guarantee that all $\succeq$ are in fact $\sim$ since for any $(z_1,\dots, z_m)$ we have $(z_1,\dots, z_m)\sim (z_1,\dots, z_m)$ so that $\sum_{i=1}^m h_i(z_i) = \sum_{i=1}^m r_i(z_i)$. This implies that $\sum_l\sum_i h_i(x_i^l)=\sum_l\sum_i r_i(y_i^l)$ giving the result. 

\section{Mixed Strategies} \label{appendix: mixed strategies}

In this section, we assume that each $S_i$ is separable metric space and that $\vert N \vert= n$. Let $\Delta(S)$ be the set of Borel probability measures on $S=\Pi_{i\in N} S_i$. 
It is well known (see e.g. \citet{grandmont1972continuity} and \citet{ok2023lipschitz}) that if a binary relation $\succeq_i$ on $\Delta(S)$ is complete, transitive, continuous and satisfies independence ($p\succeq_i q$ if and only if $\lambda p + (1-\lambda)r \succeq_i \lambda q + (1-\lambda)r$ for any $p,q,r$ and $\lambda\in [0,1]$) then there is a continuous function $u_i:S\rightarrow \mathbb{R}$ such that for any $p,q\in \Delta(S)$ we have $$p\succeq_i q \iff \int u_i dp \geq \int u_i dq.$$ $u_i$ is called a Bernoulli index representing $\succeq_i$. 
Furthermore, $u_i$ is unique up to positive affine transformations. That is, if $u_i'$ is another Bernoulli index that represents $\succeq_i$ then one can find $a\in \mathbb{R}$ and $b\in \mathbb{R}_{++}$ such that $u_i'=a+b u_i$. In this section, we seek further conditions on $\succeq_i$ so that we can find functions $u_{ij}:S_i\times S_j\rightarrow \mathbb{R}$ so that $\sum_{j\neq i} u_{ij}(s_i,s_j)$ is a Bernoulli index for $\succeq_i$.

For any $p\in \Delta(S)$, and $J\subset N$ we write $p_J$ for the marginal distribution of $p$ on $\Pi_{j\in J}S_j$. 

\begin{Prop}
    Suppose that $\succeq_i$ is complete, transitive, continuous and satisfies independence. The following two conditions are equivalent:
    \begin{enumerate}
        \item There are continuous functions $u_{ij}:S_i\times S_j\rightarrow \mathbb{R}$ so that $\sum_{j\neq i} u_{ij}(s_i,s_j)$ is a Bernoulli index for $\succeq_i$
        \item for any $p,q\in \Delta(S)$ such that $p_{\{i,j\}}=q_{\{i,j\}}$ for all $j$ then $p\sim_i q$.
    \end{enumerate}
\end{Prop}

\begin{proof}
    That $(a)$ implies $(b)$ is straightforward. For the other direction, let $u_i$ be a continuous Bernoulli index for $\succeq_i$. Fix some arbitrary $s_{-i}^*$ and arbitrary continuous functions $\bar{u}_j:S_i\rightarrow \mathbb{R}$ for each $j$ such that $$\sum_{j\neq i}\bar{u}_j(s_i)=u_i(s_i,s_{-i}^*).$$ for all $s_i$. Define $$u_{ij}(s_i,s_j)=u_i(s_i,s_j,s_{-ij}^*)-(n-2)\bar{u}_j(s_i).$$ Note that each $u_{ij}$ is continuous. We claim that if $(b)$ holds, then $u(s)=\sum_{j\neq i} u_{ij}(s_i,s_j)$ for all $s\in S$. To see this, fix $s\in S$ and let $p$ be the lottery in $\Delta(S)$ that is $s$ with probability $\frac{1}{n-1}$ and is $(s_i,s_{-i}^*)$ with probability $\frac{n-2}{n-1}$. Let $q$ be the lottery that, for each $j\neq i$, is $(s_i,s_j,s_{-ij}^{*})$ with probability $\frac{1}{n-1}$. Notice that $p_{\{i,j\}}=q_{\{i,j\}}$ for each $j\neq i$ since in both $p$ and $q$, $s_i,s_j$ appear together with probability $\frac{1}{n-1}$ and $s_i,s_j^{*}$ appear together with probability $\frac{n-2}{n-1}$. Therefore, imposing $(b)$, we have $$\frac{1}{n-1}u_i(s)+\frac{n-2}{n-1}u_i(s_i,s_{-i}^*)=  \sum_{j\neq i} \frac{1}{n-1}u_i(s_i,s_j,s_{-ij}^*).$$ Multiplying both sides by $n-1$, substituting $\sum_{j\neq i}\bar{u}_j(s_i)$ for $u_i(s_i,s_{-i}^*)$  and rearranging yields, $$u_i(s)=\sum_{j\neq i}u_i(s_i,s_j,s_{-ij}^*)-(n-2)u_{i}(s_i,s_{-i}^*)=\sum_{j\neq i} u_i(s_i,s_j,s_{-ij}^{*})-(n-2)\bar{u}_j(s_i)=\sum_{j\neq i} u_{ij}(s_i,s_j)$$ as desired.
\end{proof}

\end{document}